\documentclass[11pt,psfig,times]{article}
\usepackage[dvips]{graphics,color}
\usepackage{latexsym}
\usepackage{epsfig}
\usepackage{times}
\usepackage{amssymb}
\usepackage{enumerate}

\usepackage{tikz}

\usepackage{amsfonts}

\newcommand{\ignore}[1]{{}}

\def\p1{\phantom{1}}

\newtheorem{theorem}{Theorem}

\newtheorem{lemma}{Lemma}

\newenvironment{proof}[1][Proof]{
                  \begin{trivlist}
                  \item \upshape \bfseries #1.
                  \upshape\mdseries}
		 {\nopagebreak\hspace*{\fill}{$\Box$}\end{trivlist}}

\newcommand{\full}[1]{{}}

\begin{document}

\begin{titlepage}
\title{On the Value of Job Migration in Online Makespan Minimization}
\author{Susanne Albers\thanks{Department of Computer Science, Humboldt-Universit\"at zu Berlin, Unter den Linden 6, 10099 Berlin. {\tt albers@informatik.hu-berlin.de}}\and Matthias Hellwig\thanks{Department of Computer Science, Humboldt-Universit\"at zu Berlin, Unter den Linden 6, 10099 Berlin. {\tt mhellwig@informatik.hu-berlin.de} }}
\date{}
\maketitle

\thispagestyle{empty}

\maketitle

\begin{abstract}
Makespan minimization on identical parallel machines is a classical scheduling problem. We consider
the online scenario where a sequence of $n$ jobs has to be scheduled non-preemptively on $m$ machines
so as to minimize the maximum completion time of any job. The best competitive ratio that can be
achieved by deterministic online algorithms is in the range $[1.88,1.9201]$. Currently no randomized
online algorithm with a smaller competitiveness is known, for general $m$.

In this paper we explore the power of job migration, i.e.\ an online scheduler is allowed to perform
a limited number of job reassignments. Migration is a common technique used in theory and practice to
balance load in parallel processing environments. As our main result we settle the performance that
can be achieved by deterministic online algorithms. We develop an algorithm that
is $\alpha_m$-competitive, for any $m\geq 2$, where $\alpha_m$ is the solution of a certain equation. 
For $m=2$, $\alpha_2 = 4/3$ and  $\lim_{m\rightarrow \infty} \alpha_m  = W_{-1}(-1/e^2)/(1+ W_{-1}(-1/e^2))
\approx 1.4659$. Here $W_{-1}$ is the lower branch of the Lambert $W$ function. For $m\geq 11$,
the algorithm uses at most $7m$ migration operations. For smaller $m$, $8m$ to $10m$ operations may
be performed. We complement this result by a matching lower bound: No online algorithm that uses $o(n)$
job migrations can achieve a competitive ratio smaller than $\alpha_m$. We finally trade performance
for migrations. We give a family of algorithms that is $c$-competitive, for any $5/3\leq c \leq 2$.
For $c= 5/3$, the strategy uses at most $4m$ job migrations. For $c=1.75$, at most $2.5m$ migrations are used.
\end{abstract}

\end{titlepage}

\section{Introduction}
Makespan minimization on identical machines is a fundamental scheduling problem that has received considerable
research interest over the last forty years. Let $\sigma = J_1, \ldots, J_n$ be a sequence of jobs that has
to be scheduled non-preemptively on $m$ identical parallel machines. Each job $J_i$ is specified by a processing 
time $p_i$, $1\leq i \leq n$. The goal is to minimize the makespan, i.e.\ the maximum completion time of any
job in a schedule. In the offline setting all jobs are known in advance. In the online setting the jobs arrive
one by one. Each job $J_i$ has to be scheduled immediately on one of the machines without knowledge of any future
jobs $J_k$, $k>i$. An online algorithm $A$ is called $c$-competitive if, for any job sequence, $A$'s makespan
is at most $c$ times the optimum makespan for that sequence~\cite{ST}.   

Early work on makespan minimization studied the offline setting. Already in 1966, Graham~\cite{G} presented the
{\em List\/} scheduling algorithm that schedules each job on a least loaded machine. {\em List\/} can be used
as an offline and online strategy and achieves a performance ratio of $2-1/m$. Hochbaum and Shmoys devised a
famous polynomial time approximation scheme~\cite{HS}. More recent research, published mostly in the 1990s,
investigated the online setting. The best competitive factor that can be attained by deterministic online
algorithms is in the range $[1.88,1.9201]$. Due to this relatively high factor, compared to {\em List\/}'s ratio 
of $2-1/m$, it is interesting to consider scenarios where an online scheduler has more flexibility
to serve the job sequence. 

In this paper we investigate the impact of job migration. At any time an online algorithm may perform
{\em reassignments\/}, i.e.\ a job already scheduled on a machine may be removed and transferred 
to another machine. Process migration is a well-known and widely used technique to balance load in parallel
and distributed systems. It leads to improved processor utilization and reduced processing delays. Migration
policies have been analyzed extensively in theory and practice.

It is natural to investigate makespan minimization with job migration. In this paper we present a comprehensive
study and develop tight upper and lower bounds on the competitive ratio that can be achieved by deterministic
online algorithms. It shows that even with a very limited number of migration operations, significantly
improved performance guarantees are obtained. 

\vspace*{0.1cm}

{\bf Previous work:}
We review the most important results relevant to our work. As mentioned above, {\em List\/} is 
$(2-1/m)$-competitive. Deterministic online algorithms with a smaller competitive ratio were
presented in~\cite{A,BFKV,FW,GW,KPT}. The best algorithm currently known is 1.9201-competitive~\cite{FW}.
Lower bounds on the performance of deterministic strategies were given in~\cite{A,BKR,FKT,GRTW,R,RC}.
The best bound currently known is 1.88, for general $m$. Randomized online algorithms cannot achieve
a competitive ratio smaller than $e/(e-1)\approx 1.58$~\cite{CVW,S}. No randomized
algorithm whose competitive ratio is provably below the deterministic lower bound is currently known,
for general $m$. If job preemption is allowed, the best competitiveness of online strategies is
equal to $e/( e-1)\approx 1.58$~\cite{CVW2}.

Makespan minimization with job migration was first addressed by Aggarwal et al.~\cite{AMZ}. 
They consider an offline setting. An algorithm is given a schedule, in which all jobs are already assigned,
and a budget. The algorithm may perform job migrations up to the given budget. The authors
design strategies that perform well with respect to the best possible solution that can be constructed
with the budget. Online makespan minimization on $m=2$ machines was considered in~\cite{MLW,TY}. The best
competitiveness is 4/3. Sanders et al.~\cite{SSS} study an online setting in which before the assignment 
of each job $J_i$, jobs up to a total processing volume of $\beta p_i$ may be migrated, for some constant
$\beta$. For $\beta=4/3$, they present a 1.5-competitive algorithm. They also show a $(1+\epsilon)$-competitive
algorithm, for any $\epsilon >0$, where $\beta$ depends exponentially on $1/\epsilon$. The algorithms are
robust in that the stated competitive ratios hold after each job assignment. However in this framework,
over time, $\Omega(n)$ migrations may be performed and jobs of total processing volume  
$\beta \sum_{i=1}^n p_i$ may be moved.
 
Englert et al.~\cite{EOW} study online makespan minimization if an algorithm is given a buffer that may
be used to partially reorder the job sequence. In each step an algorithm assigns one job from the buffer
to the machines. Then the next job in $\sigma$ is admitted to the buffer. Englert et al.\ show that,
using a buffer of size $\Theta(m)$, the best competitive ratio is $W_{-1}(-1/e^2)/(1+ W_{-1}(-1/e^2))$,
where $W_{-1}$ is the Lambert $W$ function. 


{\bf Our contribution:}
We investigate online makespan minimization with limited migration. The number of job reassignments does
not depend on the length of the job sequence. We determine the exact competitiveness achieved by deterministic
algorithms, for general $m$.

In Section~\ref{sec:a1} we develop an optimal algorithm. For any $m\geq 2$, the strategy is
$\alpha_m$-competitive, where $\alpha_m$ is the solution of an equation representing load in an ideal 
machine profile for a subset of the jobs. For $m=2$, the competitive ratio is 4/3. The ratios are non-decreasing
and converge to $W_{-1}(-1/e^2)/(1+ W_{-1}(-1/e^2))\approx 1.4659$ as $m$ tends to infinity. Again, $W_{-1}$ 
is the lower branch of the Lambert $W$ function. The algorithm uses at most $(\lceil (2-\alpha_m)/(\alpha_m-1)^2\rceil+4)m$
job migrations. For $m\geq 11$, this expression is at most $7m$. For smaller machine numbers it is $8m$ to $10m$.
We note that the competitiveness of 1.4659 is considerably below the factor of roughly 1.9 obtained by
deterministic algorithms in the standard online setting. It is also below the ratio of $e/(e-1)$ attainable
if randomization or job preemption are allowed. 

In Section~\ref{sec:2} we give a matching lower bound. We show that no deterministic algorithm that uses
$o(n)$ job migrations can achieve a competitive ratio smaller than $\alpha_m$, for any $m\geq 2$. Hence
in order to beat the factor of $\alpha_m$, $\Theta(n)$ reassignments are required. Finally, 
in Section~\ref{sec:3} we trade migrations for performance. We develop a family of algorithms that is
$c$-competitive, for any constant $c$ with $5/3\leq c \leq 2$. Setting $c=5/3$ we obtain a strategy that uses
at most $4m$ job migrations. For $c=1.75$, the strategy uses no more than $2.5m$ migrations.

Our algorithms rely on a number of new ideas. All strategies classify incoming jobs into small and large
depending on a careful estimate on the optimum makespan. The algorithms consist of a job arrival
phase followed by a migration phase. The optimal algorithm, in the arrival phase, maintains a load profile
on the machines with respect to jobs that are currently small. In the migration phase, the algorithm
removes a certain number of jobs from each machine. These jobs are then rescheduled using strategies
by Graham~\cite{G,G2}. Our family of algorithms partitions the $m$ machines into two sets $A$ and $B$. In the 
arrival phase the algorithms prefer to place jobs on machines in $A$ so that machines in $B$ are available 
for later migration. In general, the main challenge in the analyses of the various algorithms is to bound the number 
of jobs that have to be migrated from each machine. 

We finally relate our contributions to some existing results. First we point out that the goal in online
makespan minimization is to construct a good schedule when jobs arrive one by one. Once the schedule is 
constructed, the processing of the jobs may start. It is not stipulated that machines start executing jobs
while other jobs of $\sigma$ still need to be scheduled. This framework is assumed in all the literature on online
makespan minimization mentioned above. Consequently it is no drawback to perform job 
migrations when the entire job sequence has arrived. Nonetheless, as for the algorithms presented in this paper, 
the machines can start processing jobs except for the up to 10 largest jobs on each machine. A second remark
is that the algorithms by Aggarwal et al.~\cite{AMZ} cannot be used to achieve good results in the online
setting. The reason is that those strategies are designed to perform well relative to the best possible
makespan attainable from an initial schedule using a given migration budget. The strategies need not
perform well compared to a globally optimal schedule. The algorithms by Aggarwal et al.\ and ours are different, 
see~\cite{AMZ}. 

On the other hand, our results exhibit similarities to those by Englert et al.~\cite{EOW} where 
a reordering buffer is given. The optimal competitive ratio of $\alpha_m$ is the solution of an equation
that also arises in~\cite{EOW}. This is due to the fact that our optimal algorithm and that in~\cite{EOW}
maintain a certain load profile on the machines. Our strategy does so w.r.t.\ jobs that are currently
small while the strategy in~\cite{EOW} considers all jobs assigned to machines. In
our framework the profile is harder to maintain because of {\em shrinking jobs\/}, i.e.\
jobs that are large at some time $t$ but small at later times $t'>t$. In the job migration phase
our algorithm reschedules jobs removed from some machines. This operation corresponds to the 
''final phase'' of the algorithm in~\cite{EOW}. However, our algorithm directly applies policies
by Graham~\cite{G,G2} while the algorithm in~\cite{EOW} computes a virtual schedule. 

In general, an interesting question is if makespan minimization
with limited migration is equivalent to makespan minimization with a bounded reordering buffer. We cannot 
prove this in the affirmative. As for the specific algorithms presented in~\cite{EOW} and in this
paper, the following relation holds. All our algorithms can be transformed into strategies with a 
reordering buffer. The competitive ratios are preserved and the number of job migrations is equal to 
the buffer size. This transformation is possible because our algorithms are {\em monotone\/}: If a job does
not have to be migrated at time $t$, assuming $\sigma$ ended at time $t$, then there is no need to migrate it
at times $t'>t$. Hence, at any time a buffer can store the candidate jobs to be migrated.
On the other hand, to the best of our knowledge, the algorithms by Englert et al.~\cite{EOW} do not translate into
strategies with job migration. All the algorithms in~\cite{EOW} use the given buffer of size $cm$, for some constant
$c$, to store the $cm$ largest jobs of the job sequence. However in our setting, a migration of the 
largest jobs does not generate good schedules. The problem are  shrinking jobs, i.e.\  jobs that are among the
largest jobs at some time $t$ but not at later times. We cannot afford to
migrate all shrinking jobs, unless we invest $\Theta(n)$ migrations. With limited job migration, 
scheduling decisions are final for almost all of the jobs. Hence the corresponding 
algorithms are more involved than in the setting with a reordering buffer.

\section{An optimal algorithm}\label{sec:a1}


For the description of the algorithm and the attained competitive ratio we define a function $f_m(\alpha)$.
Intuitively, $f_m(\alpha)$ represents accumulated normalized load in a ``perfect'' machine profile for a subset 
of the jobs. In such a profile the load ratios of the first $\lfloor m/\alpha\rfloor$ machines follow a Harmonic 
series of the form $(\alpha-1)/(m-1), \ldots, (\alpha-1)/(m-\lfloor m/\alpha\rfloor)$ while the remaining
ratios are $\alpha/m$. Summing up these ratios we obtain $f_m(\alpha)$. Formally, let 
$$f_m(\alpha) = (\alpha-1)(H_{m-1}-H_{\lceil(1-1/\alpha)m\rceil-1}) + \lceil(1-1/\alpha)m\rceil \alpha/m,$$
for any machine number $m\geq 2$ and real-valued $\alpha>1$. Here $H_k = \sum_{i=1}^k 1/i$ denotes the
$k$-th Harmonic number, for any integer $k\geq 1$. We set $H_0 = 0$. 
For any fixed $m\geq 2$, let $\alpha_m$ be the value satisfying $f_m(\alpha)=1$. 
Lemma~\ref{lem:l1} below implies that $\alpha_m$ is well-defined. The algorithm we present is exactly
$\alpha_m$-competitive. By Lemma~\ref{lem:l2}, the values $\alpha_m$ form a non-decreasing sequence.
There holds $\alpha_2 = 4/3$ and $\lim_{m\rightarrow \infty} \alpha_m = W_{-1}(-1/e^2)/(1+ W_{-1}(-1/e^2))\approx 1.4659$. This convergence was also stated by Englert et al.~\cite{EOW} but no thorough proof was presented.
The following two technical lemmas are proven in the appendix.
\begin{lemma}\label{lem:l1}
The function $f_m(\alpha)$ is continuous and strictly increasing in $\alpha$, for any integer $m\geq 2$ and real number $\alpha>1$.
There holds $f_m(1+1/(3m)) <1$ and $f_m(2) \geq 1$. 
\end{lemma}
\begin{lemma}\label{lem:l2}
The sequence $(\alpha_m)_{m\geq 2}$ is non-decreasing with $\alpha_2= 4/3$ and $\lim_{m\rightarrow \infty} \alpha_m
 = W_{-1}(-1/e^2)/(1+ W_{-1}(-1/e^2))$.
\end{lemma}

\subsection{Description of the algorithm}
Let $m\geq 2$ and $M_1, \ldots, M_m$ be the available machines. Furthermore, let $\alpha_m$ be as defined above. 
The algorithm, called {\em ALG($\alpha_m$)}, operates
in two phases, a {\em job arrival phase\/} and a {\em job migration phase\/}. In the job arrival phase all
jobs of $\sigma = J_1, \ldots, J_n$ are assigned one by one to the machines. In this phase
no job migrations are performed. Once $\sigma$ is scheduled, the job migration phase starts.
First the algorithm removes some jobs from the machines. Then these jobs are reassigned to other machines.

{\bf Job arrival phase.} In this phase {\em ALG($\alpha_m$)} classifies jobs into small and large and,
moreover, maintains a load profile with respect to the small jobs on the machines. At any time the load
of a machine is the sum of the processing times of the jobs currently assigned to it. Let {\em time $t$} be 
the time when $J_t$ has to be scheduled, $1\leq t \leq n$. 

In order to classify jobs {\em ALG($\alpha_m$)} maintains a lower bound $L_t$ on the optimum makespan. 
Let $p_t^+ = \sum_{i=1}^t p_i$ be the sum of the processing times of the first $t$ jobs. Furthermore,
for $i=1, \ldots, 2m+1$, let $p_t^i$ denote the processing time of the $i$-th largest job in 
$J_1, \ldots, J_t$, provided that $i\leq t$. More formally, if $i\leq t$, let $p_t^i$ be the processing
time of the $i$-th largest job; otherwise we set $p_t^i=0$. 
Obviously, when $t$ jobs have arrived, the optimum makespan cannot be smaller than the average load
${1\over m} p_t^+$ on the $m$ machines. Moreover, the optimum makespan cannot be smaller than $3p_t^{2m+1}$,
which is three times the processing time of $(2m+1)$-st largest job seen so far. Define
$$L_t = \max\{\textstyle{1\over m}p_t^+, 3p_t^{2m+1}\}.$$

A job $J_t$ is called {\em small\/} if $p_t \leq (\alpha_m-1)L_t$; otherwise it is {\em large}.
As the estimates $L_t$ are non-decreasing over time, a large job $J_t$ does not necessarily satisfy
$p_t > (\alpha_m-1)L_{t'}$ at times $t'>t$. Therefore we need a more refined notion of small and
large. A job $J_i$, with $i\leq t$, is {\em small at time $t$\/} if $p_i \leq (\alpha_m-1)L_t$;
otherwise it is {\em large at time $t$\/}. We introduce a final piece of notation. In the sequence
$p_t^1, \ldots, p_t^{2m}$ of the $2m$ largest processing times up to time $t$ we focus on those
that are large. More specifically, for $i=1, \ldots, 2m$, let $\hat{p}_t^i = p_t^i$ if 
$p_t^i >(\alpha_m-1)L_t$; otherwise let $\hat{p}_t^i = 0$. Define
$$\textstyle{L^*_t = {1\over m}(p_t^+ - \sum_{i=1}^{2m}\hat{p}_t^i)}.$$
Intuitively, $L_t^*$ is the average machine load ignoring jobs that are large at time $t$. Since
$\alpha_m\geq 4/3$, by Lemma~\ref{lem:l2}, and $L_t \geq 3 p_t^{2m+1}$, there can exist at most
$2m$ jobs that are large at time $t$.

\begin{figure}[t] 
{\bf Algorithm ALG($\alpha_m$):} \\[2pt]
{\em Job arrival phase.} Each $J_t$, $1\leq t \leq n$, is scheduled as follows.\\[-21pt]
\begin{itemize}
\item $J_t$ is small: Assign $J_t$ to an $M_j$ with $\ell_s(j,t)\leq \beta(j)L_t^*$.\\[-20pt]
\item $J_t$ is large: Assign $J_t$ to a least loaded machine.\\[-18pt]
\end{itemize}
{\em Job migration phase.}\\[-21pt]
\begin{itemize}
\item Job removal: Set $R:=\emptyset$. While there exists an $M_j$ with 
$\ell(j)>\max\{\beta(j)L^*,(\alpha-1)L\}$, remove the largest job from $M_j$ and add it to $R$.\\[-20pt]
\item Job reassignment: $R'= \{J_i\in R \mid p_i > (\alpha_m-1)L\}$. For $i=1, \ldots, m$, set $P_i$ contains 
$J_{r}^i$, if $i\leq |R'|$, and $J_{r}^{2m+1-i}$, if $p_r^{2m+1-i} > p_r^i/2$ and $2m+1-i\leq |R'|$. Number the sets 
in order of non-increasing total processing time. For $i=1, \ldots, m$, assign $P_i$ to a least loaded machine. 
Assign each $J_i \in R\setminus (P_1 \cup \ldots \cup P_m)$
to a least loaded machine.\\[-20pt]
\end{itemize}
\caption{The algorithm {\em ALG($\alpha_m$)}.}\label{fig:1}
\end{figure}

We describe the scheduling steps in the job arrival phase. Initially, the machines are
numbered in an arbitrary way and this numbering $M_1, \ldots, M_m$ remains fixed throughout the
execution of {\em ALG($\alpha_m$)}. As mentioned above the algorithm maintains a load profile on the 
machines as far as small jobs are concerned. Define
$$ \beta(j) = \left\{ \begin{array}{ll}
 (\alpha_m-1){m\over m-j} & \mbox{if}\  j \leq \lfloor m/\alpha_m\rfloor\\[4pt]
  \alpha_m & \mbox{otherwise.}
  \end{array}
          \right.
$$
We observe that $f_m(\alpha_m) = {1\over m}\sum_{j=1}^m \beta(j)$, taking into account that 
$m - \lfloor m/\alpha_m \rfloor = \lceil (1-1/\alpha_m)m\rceil$. For any machine $M_j$ $1\leq j \leq m$,
let $\ell(j,t)$ denote its load at time $t$ {\em before} $J_t$ is assigned to a machine. Let
$\ell_s(j,t)$ be the load caused by the jobs on $M_j$ that are small at time $t$. 
{\em ALG($\alpha_m$)} ensures that at any time $t$ there exists a machine $M_j$ satisfying
$\ell_s(j,t)\leq \beta(j)L_t^*$.

For $t=1, \ldots, n$, each $J_t$ is scheduled as follows. If $J_t$ is small, then it is scheduled on
a machine with $\ell_s(j,t)\leq \beta(j)L_t^*$. In Lemma~\ref{lem:l3} we show that such a
machine always exists. If $J_t$ is large, then it is assigned to a machine having the smallest
load among all machines. At the end of the phase let $L = L_n$ and $L^*=L_n^*$.

{\bf Job migration phase.} This phase consists of a {\em job removal step\/} followed by
a {\em job reassignment step\/}. At any time during the phase, let $\ell(j)$ denote the 
current load of $M_j$, $1\leq j\leq m$. In the removal step {\em ALG($\alpha_m$)} maintains a set
$R$ of removed jobs. Initially $R=\emptyset$. During the removal step, while there exists a machine
$M_j$ whose load $\ell(j)$ exceeds $\max\{\beta(j)L^*,(\alpha_m-1)L\}$, {\em ALG($\alpha_m$)}
removes the job with the largest processing time currently residing on $M_j$ and adds the job to
$R$. 

If $R= \emptyset$ at the end of the removal step, then {\em ALG($\alpha_m$)} terminates. If $R\neq \emptyset$,
then the reassignment step is executed. Let $R'\subseteq R$ be the subset of the jobs that
are large at the end of $\sigma$, i.e.\ whose processing time is greater than $(\alpha_m-1)L$.
Again there can exist at most $2m$ such jobs. {\em ALG($\alpha_m$)} first sorts the jobs of $R'$ in order
of non-increasing processing time; ties are broken arbitrarily. Let $J_r^i$, $1\leq i \leq |R'|$,
be the $i$-th job in this sorted sequence and $p_r^i$ be its processing time. For $i=1,\ldots,m$,
{\em ALG($\alpha_m$)} forms jobs pairs consisting of the $i$-th largest and the $(2m+1-i)$-th largest
jobs provided that the processing time of the latter job is sufficiently high. A pairing
strategy combining the $i$-th largest and the $(2m+1-i)$-th largest jobs was also used by Graham~\cite{G2}. 
Formally, {\em ALG($\alpha_m$)}
builds sets $P_1, \ldots, P_m$ that contain up to two jobs. Initially, all these sets are empty.
In a first step $J_r^i$ is assigned to $P_i$, for any $i$ with $1\leq i \leq \min\{m,|R'|\}$. In a 
second step $J_r^{2m+1-i}$ is added to $P_i$ provided that $p_r^{2m+1-i}> p_r^i/2$, i.e.\ the
processing time of $J_r^{2m+1-i}$ must be greater than half times that of $J_r^i$. This second step
is executed for any $i$ such that $1\leq i\leq m$ and $2m+1-i\leq |R'|$. For any set $P_i$, $1\leq i\leq m$, 
let $\pi_i$ be the total summed processing time of the jobs in $P_i$. {\em ALG($\alpha_m$)\/} now renumbers 
the sets in order of non-increasing $\pi_i$ values such that $\pi_1\geq \ldots \geq \pi_m$. Then, for $i=1, \ldots, m$, 
it takes the set $P_i$ and assigns the jobs of $P_i$ to a machine with the smallest current load. If $P_i$ contains two jobs,
then both are placed on the same machine. Finally, if $R\setminus (P_1\cup \ldots \cup P_m) \neq \emptyset$, then 
{\em ALG($\alpha_m$)\/}  takes care of the remaining jobs. These jobs may be scheduled in an arbitrary order. Each
job of $R\setminus (P_1\cup \ldots \cup P_m)$ is scheduled on a machine having the smallest current load. 
This concludes the description of {\em ALG($\alpha_m$)\/}. A summary in pseudo-code is given in 
Figure~\ref{fig:1}.

\begin{theorem}\label{th:1}
ALG($\alpha_m$) is $\alpha_m$-competitive and uses at most $(\lceil (2-\alpha_m)/(\alpha_m-1)^2\rceil+4)m$ 
job migrations.
\end{theorem} 
As we shall see in the analysis of {\em ALG($\alpha_m$)} in the job migration phase the algorithm has to
remove at most $\mu_m= \lceil (2-\alpha_m)/(\alpha_m-1)^2\rceil+4$ jobs from each machine. 
Table~\ref{t:1} depicts the competitive ratios $\alpha_m$ (exactly and approximately) and the migration numbers
$\mu_m$, for small values of $m$. We point out that $\alpha_m$ is a rational number, for any $m\geq 2$.

\newcommand{\tenty}[1]{\begin{minipage}{0.8cm}\vspace{2pt} \centering  $#1$  \end{minipage}}
\begin{table}
\begin{center}
\begin{tabular}{|c|c|c|c|c|c|c|c|c|c|c|}
\hline
$m$ &2 & 3 &4&5&6&7&8&9&10&11\\
\hline
\hline
$\alpha_m$& $\tenty{4\over 3}$ & $\tenty{15 \over 11}$ & $\tenty{11\over 8}$& $\tenty{125\over 89}$ & $\tenty{137\over 97}$& $\tenty{273\over 193}$ & $\tenty{586\over 411}$ & $\tenty{1863\over 1303}$ & $\tenty{5029\over 3517}$ & $\tenty{58091\over 40451}$ \\
$\approx$& & $1.3636$ & $1.375$& $1.4045$ & $1.4124$& $1.4145$ & $1.4258$ & $1.4298$ & $1.4299$ & $1.4360$ \\
\hline
$\mu_m$&10&9&9&8&8&8&8&8&8&7\\
\hline
\end{tabular}
\end{center}
\vspace*{-0.3cm}
\caption{The values of $\alpha_m$ and $\mu_m$, for small $m$.}\label{t:1}
\end{table}

\subsection{Analysis of the algorithm}
We first show that the assignment operations in the job arrival phase are well defined. A corresponding
statement was shown by Englert et al.~\cite{EOW}. The following proof is more involved because we have to 
take care of large jobs in the current schedule. 
\begin{lemma}\label{lem:l3}
At any time $t$ there exists a machine $M_j$ satisfying $\ell_s(j,t)\leq \beta(j)L_t^*$.
\end{lemma}
\begin{proof}
Suppose that there exists a time $t$, $1\leq t\leq n$, such that $\ell_s(j,t)> \beta(j)L_t^*$
holds for all $M_j$, $1\leq j\leq m$. We will derive a contradiction. 

Among the jobs $J_1, \ldots, J_t$, at most $2m$ can be large at time $t$: If there were at least
$2m+1$ such jobs, then $L_t \geq 3p_t^{2m+1} > 3(\alpha_m-1)L_t \geq L_t$ because $\alpha_m\geq 4/3$,
see Lemma~\ref{lem:l2}. Hence each of the jobs that is large at time $t$ is represented by a positive
entry in the sequence $\hat{p}_t^1, \ldots, \hat{p}_t^{2m}$. Conversely, every positive entry in this
sequence corresponds to a job that is large at time $t$ and resides on one of the $m$ machines or
is equal to $J_t$ if $J_t$ is large. Hence if $J_t$ is large, $\sum_{j=1}^m \ell(j,t) + p_t
= \sum_{j=1}^m \ell_s(j,t) + \sum_{i=1}^{2m}\hat{p}_t^i$. If  $J_t$ is small, then 
$\sum_{j=1}^m \ell(j,t) + p_t \geq \sum_{j=1}^m \ell(j,t)= \sum_{j=1}^m \ell_s(j,t) + \sum_{i=1}^{2m}\hat{p}_t^i$.
In either case
\begin{eqnarray*}
\sum_{j=1}^m \ell(j,t) + p_t &\geq& \sum_{j=1}^m \ell_s(j,t) + \sum_{i=1}^{2m}\hat{p}_t^i \
> \ \sum_{j=1}^m \beta(j) L_t^* + \sum_{i=1}^{2m}\hat{p}_t^i\\
&=& m(\alpha_m -1)L_t^*\sum_{j=1}^{\lfloor m/\alpha_m\rfloor} 1/(m-j) + (m-\lfloor m/\alpha_m\rfloor)\alpha_m L_t^*
+\sum_{i=1}^{2m}\hat{p}_t^i.
\end{eqnarray*}
Taking into account that $m- \lfloor m/\alpha\rfloor = \lceil(1-1/\alpha_m)m \rceil$ and
that $f_m(\alpha_m) = 1$, we obtain  
\begin{eqnarray*}
\sum_{j=1}^m \ell(j,t) + p_t &>& mL_t^* ((\alpha_m-1)(H_{m-1}-H_{\lceil(1-1/\alpha_m)m\rceil-1}) + \lceil(1-1/\alpha_m)m\rceil \alpha_m/m) + \sum_{i=1}^{2m}\hat{p}_t^i\\
&=&mL_t^*f_m(\alpha_m) + \sum_{i=1}^{2m}\hat{p}_t^i  
= m(1/m\sum_{i=1}^t p_t - 1/m \sum_{i=1}^{2m}\hat{p}_t^i) + \sum_{i=1}^{2m}\hat{p}_t^i \ = \  \sum_{i=1}^t p_i.
\end{eqnarray*}
This contradicts the fact that $\sum_{j=1}^m \ell(j,t) + p_t$ is equal to the total processing 
time $\sum_{i=1}^t p_i$  of $J_1, \ldots, J_t$.
\end{proof}
We next analyze the job migration phase.
\begin{lemma}\label{lem:r1}
In the job removal step ALG($\alpha_m$) removes at most $\lceil (2-\alpha_m)/(\alpha_m-1)^2\rceil +4$ jobs
from each of the machines.
\end{lemma}
\begin{proof}
Consider any $M_j$, with $1\leq j \leq m$. We show that it suffices to remove at most 
$\lceil (2-\alpha_m)/(\alpha_m-1)^2\rceil +4$ jobs so that $M_j$'s resulting load is upper bounded by 
$\max\{\beta(j)L^*,(\alpha_m-1)L\}$. Since {\em ALG($\alpha_m$)} always removes the largest jobs the
lemma follows.

Let time $n+1$ be the time when the entire job sequence $\sigma$ is scheduled and the job migration phase
with the removal step starts. A job $J_i$, with $1\leq i \leq n$, is {\em small at time $n+1$\/} if 
$p_i \leq (\alpha_m-1)L$; otherwise it is {\em large at time $n+1$}. 
Since $L=L_n$ any job that is small (large) at time $n+1$ is also small (large) at time $n$. 
Let $\ell(j,n+1)$ be the load of $M_j$ at time $n+1$. Similarly, $\ell_s(j,n+1)$ is $M_j$'s load
consisting of the jobs that are small at time $n+1$. Throughout the proof let 
$k:= \lceil (2-\alpha_m)/(\alpha_m-1)^2\rceil$.

First assume $\ell_s(j,n+1) \leq \beta(j)L^*$. If at time $n+1$ machine $M_j$ does not contain any
jobs that are large at time $n+1$, then $\ell(j,n+1) = \ell_s(j,n+1) \leq  \beta(j)L^*$. In this case
no job has to be removed and we are done. If $M_j$ does contain jobs that are large at time $n+1$, then
it suffices to remove these jobs. Let time $l$ be the last time when a job $J_l$ that is large at time
$n+1$ was assigned to $M_j$. Since $L_l \leq L$, $J_l$ was also large at time $l$ and hence it was
assigned to a least loaded machine. This implies that prior to the assignment of $J_l$, $M_j$
has a load of at most $p_l^+/m\leq L_l\leq L$. Hence it could contain at most $1/(\alpha_m-1)$ jobs that are large
at time $n+1$ because any such job has a processing time greater than $(\alpha_m-1)L$. Hence at most
$1/(\alpha_m-1)+1$ jobs have to be removed from $M_j$, and the latter expression is upper bounded by
$k+4$.

Next assume $\ell_s(j,n+1) >\beta(j)L^*$. If $\ell_s(j,n) \leq \beta(j)L^*= \beta(j)L_n^*$, 
then $J_n$ was assigned to $M_j$. In this case it suffices to remove $J_n$ and, as in the previous
case, at most $1/(\alpha_m-1)+1$ jobs that are large at time $n+1$. Again $1/(\alpha_m-1)+2\leq k+4$.

In the remainder of this proof we consider the case that $\ell_s(j,n+1) >\beta(j)L^*$ and
$\ell_s(j,n) > \beta(j)L_n^*$. Let $t^*$ be the earliest time such that $\ell_s(j,t) > \beta(j)L_t^*$
holds for all times $t^*\leq t \leq n$. We have $t^*\geq 2$ because $\ell_s(j,1)=0\leq \beta(j)L_1^*$.
Hence time $t^*-1$ exists. We partition the jobs residing on $M_j$ at time $n+1$ into three sets. 
Set $T_1$ is the set of jobs that were assigned to $M_j$ at or before time $t^*-1$ and are small at time $t^*-1$. 
Set $T_2$ contains the jobs that were assigned to $M_j$ at or before time $t^*-1$ and are large at 
time $t^*-1$. Finally $T_3$ is the set of jobs assigned to $M_j$ at or after time $t^*$. We show 
a number of claims that we will use in the further proof.
\begin{enumerate}[\bf { Claim~\ref{lem:r1}.}1.]
\item Each job in $T_2\cup T_3$ is large at the time it is assigned to $M_j$.\label{c:r11}
\item There holds $\sum_{J_i\in T_1\setminus\{J_l\}} p_i \leq \beta(j)L_{t^*-1}^*$, where $J_l$ is the 
 job of $T_1$ that was assigned last to $M_j$.\label{c:r12}
\item There holds $|T_2| \leq 3$.\label{c:r13}
\item For any $J_l\in T_3$, $M_j$'s load immediately before the assignment of $J_l$ is at most $L_l$.\label{c:r14}
\item Let $J_l\in T_3$ be the last job assigned to $M_j$. If $M_j$ contains at least $k$ jobs, 
 different from $J_l$, each having a processing time of at least $(\alpha_m-1)^2L$, then it suffices
 to remove these $k$ jobs and $J_l$ such that $M_j$'s resulting load is upper bounded by
 $(\alpha_m-1)L$. \label{c:r15}
\item If there exists a $J_l\in T_3$ with $p_l<(\alpha_m-1)^2L$, then $M_j$'s load immediately before
the assignment of $J_l$ is at most $(\alpha_m-1)L$.\label{c:r16}  
\end{enumerate}

{\em Proof of Claim~\ref{lem:r1}.\ref{c:r11}.} The jobs of $T_2$ are large at time $t^*-1$ and hence
at the time they were assigned to $M_j$. By the definition of $t^*$, $\ell_s(j,t)> \beta(j)L_t^*$
for any $t^*\leq t \leq n$. Hence {\em ALG($\alpha_m$)} does not assign small jobs to $M_j$ at or 
after time $t^*$.

{\em Proof of Claim~\ref{lem:r1}.\ref{c:r12}.} All jobs of $T_1\setminus\{J_l\}$ are small at time
$t^*-1$ and their total processing time is at most $\ell_s(j,t^*-1)$. In fact, their total processing
time is equal to $\ell_s(j,t^*-1)$ if $l=t^*-1$. By the definition of $t^*$, $\ell_s(j,t^*-1)\leq
\beta(j)L_{t^*-1}^*$.

{\em Proof of Claim~\ref{lem:r1}.\ref{c:r13}.} We show that for any time $t$, $1\leq t\leq n$, when
$J_t$ has been placed on a machine, $M_j$ can contain at most three jobs that are large at time $t$. 
The claim then follows by considering $t^*-1$. Suppose that when $J_t$ has been scheduled, $M_j$ contained
more than three jobs that are large as time $t$. Among these jobs let $J_l$ be the one that was
assigned last to $M_j$. Immediately before the assignment of $J_l$ machine $M_j$ had a load greater
than $L_l$ because the total processing time of three large jobs is greater than $3(\alpha_m-1)L_t\geq 
3(\alpha_m-1)L_l\geq L_l$ since $\alpha_m\geq 4/3$, see Lemma~\ref{lem:l2}. This contradicts the
fact that $J_l$ is placed on a least loaded machine, which has a load of at most $p_l^+/m \leq 
L_l$. 

{\em Proof of Claim~\ref{lem:r1}.\ref{c:r14}.} By Claim~\ref{lem:r1}.\ref{c:r11} $J_l$ is large at time 
$l$ and hence is assigned to a least loaded machine, which has a load of at most $p_l^+/m \leq L_l$.

{\em Proof of Claim~\ref{lem:r1}.\ref{c:r15}.} Claim~\ref{lem:r1}.\ref{c:r14} implies that
immediately before the assignment of $J_l$ machine $M_j$ has a load of at most $L_l\leq L$. If $M_j$
contains at least $k$ jobs, different from $J_l$, with a processing time of at least 
$(\alpha_m-1)^2L$, then the removal of these $k$ jobs and $J_l$ from $M_j$ leads to a machine
load of at most $L - k(\alpha_m-1)^2 L \leq L -\lceil (2-\alpha_m)/(\alpha_m-1)^2\rceil(\alpha_m-1)^2 L
\leq (\alpha_m-1)L$, as desired.

{\em Proof of Claim~\ref{lem:r1}.\ref{c:r16}.} By Claim~\ref{lem:r1}.\ref{c:r11} $J_l$ is large 
at time $l$ and hence $p_l> (\alpha_m-1)L_l$. Since $p_l<(\alpha_m-1)^2L$, it follows $L_l<(\alpha_m-1)L$.
By Claim~\ref{lem:r1}.\ref{c:r14}, $M_j$'s load prior to the assignment of $J_l$ is at most $L_l$ and 
hence at most $(\alpha_m-1)L$.

We now finish the proof of the lemma and distinguish two cases depending on the cardinality of $T_2\cup T_3$.

{\bf Case 1:} If $|T_2\cup T_3| <k+4$, then by Claim~\ref{lem:r1}.\ref{c:r12} it suffices to remove the
jobs of $T_2\cup T_3$ and the last job of $T_1$ assigned to $M_j$.

{\bf Case 2:} Suppose $|T_2\cup T_3| \geq k+4$. By Claim~\ref{lem:r1}.\ref{c:r13}, $|T_2|\leq 3$ and hence 
$|T_3|\geq k+1$. Among the jobs of $T_3$ consider the last $k+1$ ones assigned to $M_j$. If each of them
has a processing time of at least $(\alpha_m-1)^2L$, then Claim~\ref{lem:r1}.\ref{c:r15} ensures that
it suffices to remove these $k+1$ jobs. If one of them, say $J_l$,  has a processing time smaller than 
$(\alpha_m-1)^2L$, then by Claim~\ref{lem:r1}.\ref{c:r16} $M_j$'s load prior to the assignment of $J_l$
is at most $(\alpha_m-1)L$. Again it suffices to remove these $k+1$ jobs from $M_j$.
\end{proof}

After the job removal step each machine $M_j$, $1\leq j\leq m$, has a load of at most 
$\max\{\beta(j)L^*,(\alpha_m-1)L\}$. We first observe that this load is at most $\alpha_mL$.
If $(\alpha_m-1)L\geq \beta(j)L^* $, there is nothing to show. We evaluate $\beta(j)L^*$. 
If $j> \lfloor m/\alpha_m\rfloor$, then $\beta(j)=\alpha_m$ and $\beta(j)L^* = \alpha_m L^*
\leq \alpha_m L$. If $j\leq \lfloor m/\alpha_m\rfloor$, then $\beta(j) = (\alpha_m-1)m/(m-j) 
\leq (\alpha_m-1)m/(m-\lfloor m/\alpha_m\rfloor) = (\alpha_m-1)m/\lceil (1-1/\alpha_m)m\rceil)\leq \alpha_m$
and thus $\beta(j)L^* \leq \alpha_m L$. Hence $M_j$'s load is upper bounded by $\alpha_m {\mathit OPT}$, where
${\mathit OPT}$ denotes the value of the optimum makespan for the job sequence $\sigma$. The following lemma
ensures that after the reassignment step, each machine still has a load of at most $\alpha_m {\mathit OPT}$. 

\begin{lemma}
After the reassignment step each machine $M_j$, $1\leq j\leq m$, has a load of at most $\alpha_m {\mathit OPT}$.
\end{lemma}
\begin{proof}
We show that all scheduling operations in the reassignment step preserve a load of at most $\alpha_m {\mathit OPT}$ on
each of the machines. We first consider the assignment of the sets $P_1, \ldots, P_m$. Suppose that these
sets are already sorted in order of non-increasing total processing times, i.e.\ $\pi_1\geq \ldots \geq \pi_m$.
We first argue that $\pi_1$ and hence any $\pi_i$, $1\leq i\leq m$, is upper bounded by $\mathit{OPT}$.
If $P_1$ contains at most one job, there is nothing to show because $\mathit{OPT}$ cannot be smaller than the 
processing time of any job in $\sigma$. Assume that $P_1$ contains two jobs. Then it consists of jobs 
$J_r^{i_1}$ and $J_r^{2m+1-i_1}$, for some $i_1$ with
$1\leq i_1 \leq m$. Since the two jobs are paired there holds $p_r^{2m+1-i_1} > p_r^{i_1}/2$ and
hence $p_r^{2m+1-i_1} > \pi_1/3$. Let $\mathit{OPT'}$ denote the optimum makespan for the job sequence
$J_r^1, \ldots, J_r^{2m+1-i_1}$. Since $J_r^{i_1}$ and $J_r^{2m+1-i_1}$ are paired, jobs $J_r^{i}$ and $J_r^{2m+1-i}$
are also paired, for any $i_1 <i \leq m$, because $p_r^{2m+1-i} \geq p_r^{2m+1-i_1} > p_r^{i_1}/2 \geq p_r^{i}/2$. Further, the sets $P_1, \dots, P_m$ contain all jobs $J_r^1,\dots, J_r^{i^*-1}$, and none of these was paired. Thus the sets $P_1, \ldots, P_m$ contain all the jobs $J_r^1, \ldots, J_r^{2m+1-i_1}$, which
implies $\pi_1\geq \mathit{OPT'}$ and $p_r^{2m+1-i_1} > \mathit{OPT'}/3$. It follows
$p_r^i > \mathit{OPT'}/3$, for all $i$ with $1\leq i \leq 2m+1-i_1$. Graham~\cite{G2} showed that given a sequence
of up to $2m$ jobs, each having a processing time greater than a third times the optimum makespan, an
optimal schedule is obtained by repeatedly pairing the $i$-th largest and $(2m+1-i)$-th largest jobs of the sequence. 
This is exactly the assignment computed by {\em ALG($\alpha_m$)} for $J_r^1, \ldots, J_r^{2m+1-i_1}$. 
We conclude $\pi_1= \mathit{OPT'}$ and $\pi_1\leq \mathit{OPT}$.

A final observation is that each job
of $R'$ that is not contained in $P_1\cup\ldots \cup P_m$ has a processing time of at most
$\mathit{OPT}/3$. A job in $R'\setminus (P_1\cup\ldots \cup P_m)$ is equal to a job $J_r^{2m+1-i_0}$, with
$1\leq i_0< i_1$. Since $J_r^{2m+1-i_0}$ is not paired with $J_r^{i_0}$, there holds 
$p_r^{2m+1-i_0} \leq p_r^{i_0}/2$. Assume that $p_r^{2m+1-i_0} > \mathit{OPT}/3$. Then $p_r^{2m+1-i_0}$
is greater than a third times the optimum makespan for the jobs $J_r^1, \ldots, J_r^{2m+1-i_0}$.
Using again the results by Graham~\cite{G2}, we obtain that an optimal schedule for the latter job
sequence in obtained by repeatedly pairing $J_r^i$ with $J_r^{2m+1-i}$. However, since 
$p_r^{2m+1-i_0} \leq p_r^{i_0}/2$, the processing time $p_r^{2m+1-i_0}$ is at most a third times the
resulting optimum makespan for $J_r^1, \ldots, J_r^{2m+1-i_0}$. Hence $p_r^{2m+1-i_0}$ is at most
a third times $\mathit{OPT}$, which is a contradiction. 

Next we compare the processing time of the jobs of $P_1\cup \ldots\cup P_m$ to $\sum_{i=1}^{2m}\hat{p}_n^i$.
Set $R'$ contains the jobs of $R$ that are large at time $n+1$. There exist at most $2m$ jobs that are large 
at time $n+1$ and hence the processing time of each job in $R'$ is represented by a positive entry in the sequence
$\hat{p}_n^1, \ldots, \hat{p}_n^{2m}$. It follows that the total processing time of the jobs in $R'$ 
and hence the total processing time of the jobs in $P_1\cup \ldots\cup P_m$ is at most 
$\sum_{i=1}^{2m}\hat{p}_n^i$. Recall that $\pi_1\geq \ldots \geq \pi_m$. Then, for any $j$ with $1\leq j\leq m$, the
product $j\pi_j$ is upper bounded by the total processing time of $P_1\cup \ldots\cup P_m$ and hence
$j\pi_j\leq \sum_{i=1}^{2m}\hat{p}_n^i$.

Now consider the assignment of the sets $P_1, \ldots, P_m$ to the machines. Each set is assigned to
a least loaded machine. Hence when $P_j$, $1\leq j \leq m$, is scheduled, it is assigned to a machine
whose current load is at most $\max\{\beta(j)L^*,(\alpha_m-1)L\}$. If the load is at most $(\alpha_m-1)L$,
then the machine's load after the assignment is at most $(\alpha_m-1)L+\pi_j \leq (\alpha_m-1)L+\mathit{OPT}
\leq \alpha_m\mathit{OPT}$. If the current load is only upper bounded by $\beta(j)L^*$, then we distinguish
two cases. 

If $j\leq \lfloor m/\alpha_m\rfloor$, then $j\leq m/\alpha_m$, which is equivalent to 
$m/(m-j) \leq \alpha_m/(\alpha_m-1)$. The resulting machine load is at most
$$ \beta(j)L^*+\pi_j = (\alpha_m-1){m\over m-j}({1\over m}\sum_{i=1}^n p_i - {1\over m}\sum_{i=1}^{2m}\hat{p}_t^j) +\pi_j \leq (\alpha_m-1){1\over m-j}(mL - j\pi_j) +\pi_j.$$
The last inequality follows because, as argued above, $j\pi_j \leq \sum_{i=1}^{2m}\hat{p}_t^i$. It 
follows that the machine load is upper bounded by 
$$(\alpha_m-1)\textstyle{1\over m-j}(mL - m\pi_j) +\alpha_m\pi_j 
\leq \alpha_m(L-\pi_j) + \alpha_m\pi_j= \alpha_mL.$$
The last inequality holds because $m/(m-j) \leq \alpha_m/(\alpha_m-1)$, as mentioned above.

If $j> \lfloor m/\alpha_m\rfloor$, then $j\geq m/\alpha_m$ because $j$ is integral. In this case the
machine load is upper bounded by 
$$\beta(j)L^*+\pi_j = \alpha_m(\sum_{i=1}^n p_i - \sum_{i=1}^{2m}\hat{p}_t^i)/m +\pi_j
\leq \alpha_m(\sum_{i=1}^n p_i- j\pi_j)/m+\pi_j \ \leq \ \alpha_m L,$$
because $j\alpha_m\geq m$.

Finally we consider the jobs $R\setminus (P_1\cup\ldots \cup P_m)$. Each job of $R\setminus R'$
has a processing time of at most $(\alpha_m-1)L$. As argued above, each job of 
$R'\setminus (P_1\cup\ldots \cup P_m)$ has a processing time of at most ${\mathit OPT}/3$, which is
upper bounded by $(\alpha_m-1)\mathit{OPT}$ since $\alpha_m\geq 4/3$. Hence each job of
$R\setminus (P_1\cup\ldots \cup P_m)$ has a processing time of at most $(\alpha_m-1)\mathit{OPT}$.
Each of these jobs is scheduled on a least loaded machine and thus after the assignment 
the corresponding machine has a load of at most $\mathit{OPT}+ (\alpha_m-1)\mathit{OPT}
\leq \alpha_m \mathit{OPT}$.
\end{proof}
The proof of Theorem~\ref{th:1} is complete.

\section{A lower bound}\label{sec:2}
We present a lower bound showing that {\em ALG($\alpha_m$)} is optimal.
\begin{theorem}
Let $m\geq 2$. No deterministic online algorithm can achieve a competitive ratio smaller than
$\alpha_m$ if $o(n)$ job migrations are allowed.
\end{theorem}
\begin{proof}
Let $A$ be any deterministic online algorithm that is allowed to use up to $g(n)$ job migrations on
a job sequence of length $n$. Suppose that $A$ achieves a competitive ratio smaller than $\alpha_m$.
We will derive a contradiction.

Choose an $\epsilon >0$ such that $A$ has a competitive ratio strictly smaller than $\alpha_m-\epsilon$. 
Let $\epsilon'=\epsilon/3$. Since $g(n)=o(n)$ there exists an $n_0$ such that $g(n)/n \leq \epsilon'/(2m)$, for all $n\geq n_0$. Hence there exists an $n_0$ such that $g(n+m)/(n+m) \leq \epsilon'/(2m)$, for all 
$n\geq \max\{m,n_0\}$. Let $n'$, with $n'\geq \max\{m,n_0\}$, be the smallest integer multiple of $m$.  We have
$g(n'+m)/n'\leq \epsilon'/m$ because $n'+m\leq 2n'$. An adversary constructs a job sequence consisting of
$n'+m$ jobs. Let $p_1= m/n'$. By our choice of $n'$, there holds $p_1\leq \epsilon'/g(n'+m)$.
The following adversarial sequence is similar to that used by Englert et al.~\cite{EOW}. However,
here we have to ensure that in migrating $o(n)$ jobs, an online algorithm cannot benefit much.

First the adversary presents $n'$ jobs of processing time $p_1$. We will refer to them as $p_1$-jobs.
If after the assignment of these jobs $A$ has a machine $M_j$, $1\leq j \leq m$, whose load is
at least $\alpha_m$, then the adversary presents $m$ jobs of processing time $p_2=\epsilon'/m$. 
Using job migration, $A$ can remove at most $g(n'+m)$ $p_1$-jobs from $M_j$. Since $g(n'+m)p_1\leq \epsilon'$,
after job migration $M_j$ still has a load of at least $\alpha_m-\epsilon$. On the other hand the
optimal makespan is $1+\epsilon'/m$. In an optimal assignment each machine contains $n'/m$ $p_1$-jobs
and one $p_2$-job. The ratio $(\alpha_m-\epsilon')/(1+\epsilon'/m)$ is at least $\alpha_m-\epsilon$ by
our choice of $\epsilon'$ and the fact that $\alpha_m\leq2$, see Lemma~\ref{lem:l1}. We obtain
a contradiction.  

In the following we study the case that after the assignment of the $p_1$-jobs each machine in $A$'s
schedule has a load strictly smaller than $\alpha_m$. We number the machines in order of non-decreasing 
load such that $\ell(1)\leq \ldots\leq \ell(m)$. Here $\ell(j)$ denotes the load of $M_j$ after the $p_1$-jobs 
have arrived, $1\leq j\leq m$. For $j=1, \ldots, m-1$, define $\beta(j) = (\alpha_m-1)m/(m-j)$. 
We first argue that there must exist a machine $M_j$, $1\leq j\leq m-1$, in $A$'s schedule whose load
is at least $\beta(j)$. Suppose that each machine $M_j$, $1\leq j\leq m-1$, had a load strictly smaller
than $\beta(j)$. By Lemma~\ref{lem:l1}, $\alpha_m>1$ and hence $\lceil(1-1/\alpha_m)m\rceil\geq 1$. 
Consider the $\lceil(1-1/\alpha_m)m\rceil$ machines with the highest load in $A$'s schedule. Each
of these machines has a load strictly smaller than $\alpha_m$. The remaining machines have a load 
strictly smaller than $\beta(j) = (\alpha_m-1)m/(m-j)$, for $j=1,\ldots, m - \lceil(1-1/\alpha_m)m\rceil$.
We conclude that after the arrival of the $p_1$-jobs the total load on the machines is strictly smaller
than
\begin{eqnarray*}
& & (\alpha_m-1)m \sum_{j=1}^{m-\lceil(1-1/\alpha_m)m\rceil} {1\over m-j}  + \lceil(1-1/\alpha_m)m\rceil \alpha_m\\
&=& m((\alpha_m-1)(H_{m-1}-H_{\lceil(1-1/\alpha_m)m\rceil-1}) + \lceil(1-1/\alpha_m)m\rceil \alpha_m/m) 
= mf_m(\alpha_m) = m.
\end{eqnarray*}
The last equation holds because $f_m(\alpha_m)=1$, by the choice of $\alpha_m$. We obtain a contradiction
to the fact that after the arrival of the $p_1$-jobs a total load of exactly $m$ resides on the machines.

Let $M_{j_0}$, with $1\leq j_0\leq m-1$, be a machine whose load is at least $\beta(j_0)$. Since $A$'s machines
are numbered in order of non-decreasing load there exist at most $j_0-1$ machines having a smaller load 
than $\beta(j_0)$.
The adversary presents $j_0$ jobs of processing time $p_2=m/(m-j_0)$. Using job migration $A$ can remove
at most $g(n'+m)$ $p_1$-jobs from any of the machines, thereby reducing the load by at most $\epsilon'$. Hence in 
$A$'s final schedule there exists a machine having a load of a least $\beta(j_0)+m/(m-j_0)-\epsilon'$. 
This holds true if the $p_2$-jobs reside on different machines. If there exists a machine containing
two $p_2$-jobs, then its load is at least $2m/(m-j_0)\geq (\alpha_m-1)m/(m-j_0) + m/(m-j_0) =
\beta(j_0) + m/(m-j_0)$ as desired. The inequality holds because $\alpha_m\leq 2$, by 
Lemma~\ref{lem:l1}. Hence $A$'s makespan is at least $\beta(j_0)+m/(m-j_0)-\epsilon'$.

The optimum makespan for the job sequence is upper bounded by $m/(m-j_0)+\epsilon'$. In an optimal 
schedule the $j_0$ $p_2$-jobs are assigned to different machines. The $n'$ $p_1$-jobs are distributed
evenly among the remaining $m-j_0$ machines. If $n'$ is an integer multiple of $m-j_0$, then the load
on each of these $m-j_0$ machines is exactly $n'p_1/(m-j_0)= m/(m-j_0)$, which is exactly equal to
the processing time of a $p_2$-job. If $n'$ is not divisible by $m-j_0$, then the maximum load on
any of these $m-j_0$ machines cannot be higher than $m/(m-j_0)+p_1 \leq m/(m-j_0) + \epsilon'/g(n'+m)
\leq m/(m-j_0) + \epsilon'$.

Dividing the lower bound on $A$'s makespan by the upper bound on the optimum makespan we obtain
$(\alpha_m m/(m-j_0) - \epsilon')/(m/(m-j_0) + \epsilon') \geq (\alpha_m-\epsilon')/(1+\epsilon')\geq  
\alpha_m-\epsilon$.
The last inequality holds because $\epsilon'=\epsilon/3$ and $\alpha_m\leq2$, see Lemma~\ref{lem:l1}.
We obtain a contradiction to the assumption that $A$'s competitiveness is strictly smaller than 
$\alpha_m-\epsilon$.
\end{proof}

\section{Algorithms using fewer migrations}\label{sec:3}

We present a family of algorithms {\em ALG$(c)$\/} that uses a smaller number of job migrations. 
We first describe the family and then analyze its performance.

\subsection{Description of {\em ALG$(c)$}}
{\em ALG$(c)$\/} is defined for any constant $c$ with $5/3\leq c \leq 2$, where $c$ is the targeted competitive ratio.
An important feature of {\em ALG$(c)$\/} is that it partitions the machines 
$M_1, \ldots, M_m$ into two sets $A=\{M_1, \ldots, M_{\lfloor m/2\rfloor}\}$ and 
$B=\{M_{\lceil m/2\rceil}, \ldots, M_m\}$ of roughly equal size. In a job arrival phase the jobs 
are preferably assigned to machines in $A$, provided that their load it not too high. In the
job migration phase, jobs are mostly migrated from machines of $A$ (preferably to machines in
$B$) and this policy will allow us to achieve a smaller number of migrations.
Setting $c=5/3$ we obtain an algorithm {\em ALG$(5/3)$\/} that is $5/3$-competitive using $4m$ migrations. 
For $c=1.75$ the resulting algorithm {\em ALG$(1.75)$\/} is $1.75$-competitive and uses at most $2.5m$ 
migrations. In the following let $5/3\leq c \leq 2$.

\vspace*{0.1cm}

{\bf Algorithm ALG$(c)$:} 
{\bf Job arrival phase.} At any time $t$ {\em ALG$(c)$\/} maintains a lower bound $L_t$ on the 
optimum makespan, which is defined as 
$\textstyle{L_t =  \max\{{1\over m}p_t^+,p_t^1,2p_t^{m+1}\}}.$
Here we use the same notation as in Section~\ref{sec:a1}. Recall that $p_t^1$ and $p_t^{m+1}$ are
the processing times of the largest and $(m+1)$-st largest jobs in $J_1, \ldots, J_t$, respectively. 
A job $J_t$ is {\em small} if $p_t\leq (2c-3)L_t$; otherwise it is {\em large}. A job $J_i$,
with $i\leq t$, is {\em small at time $t$\/} if  $p_i\leq (2c-3)L_t$. For any machine $M_j$ and any time
$t$, $\ell(j,t)$ is $M_j$'s load immediately before $J_t$ is assigned and $\ell_s(j,t)$ is its
load consisting of the jobs that are small at time $t$.

Any job $J_t$, $1\leq t\leq n$, is processed as follows. If $J_t$ is small, then {\em ALG$(c)$\/}
checks if there is a machine in $A$ whose load value $\ell_s(j,t)$ is at most $(c-1)L_t$. 
If this is the case, then among the machines in $A$ with this property, $J_t$ is assigned to one
having the smallest $\ell_s(j,t)$ value. If there is no such machine in $A$, then $J_t$ is
assigned to a least loaded machine in $B$. If $J_t$ is large, then {\em ALG$(c)$\/} checks if there
is machine in $A$ whose load value $\ell(j,t)$ is at most $(3-c)L_t$. If this is the case,
then $J_t$ is scheduled on a least loaded machine in $A$. Otherwise $J_t$ is assigned to a least
loaded machine in $B$. At the end of the phase let $L= L_n$.

{\bf Job migration phase.} At any time during the phase let $\ell(j)$ denote the current load of
$M_j$, $1\leq j\leq m$. We first describe the job removal step. For any machine $M_j\in B$, 
{\em ALG$(c)$\/} removes the largest job from that machine. Furthermore, while there exists
a machine $M_j\in A$ whose current load exceeds $(c-1)L$, {\em ALG$(c)$\/} removes the
largest job from the machine. Let $R$ be the set of all removed jobs. In the job reassignment
step {\em ALG$(c)$\/} first sorts the jobs in order of non-increasing processing times. For
any $i$, $1\leq i \leq |R|$, let $J_r^i$ be the $i$-th largest job in this sequence, and let
$p_r^i$ be the corresponding processing time. For $i=1, \ldots, |R|$, $J_r^i$ is scheduled
as follows. If there exists a machine $M_j\in B$ such that $\ell(j)+p_r^i\leq cL$, i.e.\
$J_r^i$ can be placed on $M_j$ without exceeding a makespan of $cL$, then $J_r^i$ 
is assigned to this machine. Otherwise the job is scheduled on a least loaded machine in $A$.
A pseudo-code description of {\em ALG$(c)$\/} is given in Figure~\ref{fig:2}.

\begin{figure}[ht] 
{\bf Algorithm ALG$(c)$:} Let $5/3\leq c \leq 2$.\\[2pt]
{\em Job arrival phase.} Each $J_t$, $1\leq t \leq n$, is scheduled as follows.\\[-21pt]
\begin{itemize}
\item $J_t$ is small: Let $A'=\{M_j\in A \mid \ell_s(j,t)\leq (c-1)L_t\}$. If $A'\neq \emptyset$, then
assign $J_t$ to a machine $M_j\in A'$ having the smallest $\ell_s(j,t)$ value. Otherwise assign $J_t$ to
a least loaded machine $M_j\in B$.\\[-20pt]
\item $J_t$ is large: If there is an $M_j\in A$ with $\ell(j,t)\leq (3-c)L_t$, then assign $J_t$ to a least
loaded machine in $A$. Otherwise assign $J_t$ to a least loaded machine in $B$.\\[-18pt]
\end{itemize}
{\em Job migration phase.}\\[-21pt]
\begin{itemize}
\item Job removal: Set $R:=\emptyset$. For any $M_j\in B$, remove the largest job from $M_j$ and add
it to $R$. While there exists an $M_j\in A$ with 
$\ell(j)>(c-1)L$, remove the largest job from $M_j$ and add it to $R$.\\[-20pt]
\item Job reassignment: Sort the jobs of $R$ in order of non-increasing processing time. 
For $i=1, \ldots, |R|$, schedule $J_r^i$ as follows. If there is an 
$M_j\in B$ with $\ell(j) + p_r^i\leq cL$, then assign $J_r^i$ to $M_j$. Otherwise assign
it to a least loaded machine in $A$.\\[-20pt]
\end{itemize}
\caption{The algorithm {\em ALG$(c)$}.}\label{fig:2}\vspace*{-0.1cm}
\end{figure}
\begin{theorem}\label{th:family}
ALG$(c)$ is $c$-competitive, for any constant $c$ with $5/3\leq c\leq 2$. 
\end{theorem} 
The proof of the above theorem is presented in Section~\ref{sec:family21}.
In order to obtain good upper bounds on the number of job migrations, we focus on specific values of
$c$. First, set $c=5/3$. In {\em ALG(5/3)} a job $J_t$ is small if $p_t \leq 1/3\cdot L_t$. In the arrival
phase a small job is assigned to a machine in $A$ if there exists a machine in this set whose load
consisting of jobs that are currently small is at most $2/3\cdot L_t$. A large job is assigned
to a machine in $A$ if there exists a machine in this set whose load is at most $4/3 L_t$. 
\begin{theorem}\label{th:3}
ALG(5/3) is ${5\over 3}$-competitive and uses at most $4m$ job migrations.
\end{theorem} 
In fact, for any $c$ with $5/3\leq c \leq 2$, {\em ALG$(c)$} uses at most $4m$ job migrations. 
Finally, let $c=1.75$. In {\em ALG(1.75)} a job $J_t$ is small if $p_t \leq 0.5\cdot L_t$. In the arrival
phase a small job is assigned to a machine in $A$ if there is a machine in this set whose load
consisting of jobs that are currently small is no more than $0.75 L_t$. A large job is assigned
to a machine in $A$ if there exists a machine in this set whose load is at most $1.25 L_t$. 
\begin{theorem}\label{th:4}
ALG(1.75) is $1.75$-competitive and uses at most $2.5m$ job migrations.
\end{theorem} 
Again, for any $c$ with $1.75\leq c \leq 2$, {\em ALG$(c)$} uses at most $2.5m$ job migrations. The proofs of 
Theorems~\ref{th:3} and \ref{th:4} are contained in Section~\ref{sec:family22}.  

\subsection{Analysis of {\em ALG$(c)$}}

In this section we analyze {\em ALG$(c)$}, for any $c$ with $5/3\leq c \leq 2$, and prove
Theorems~\ref{th:family}, \ref{th:3} and \ref{th:3}. We first determine the competitive
ratio of {\em ALG$(c)$} and then bound the number of job migrations performed for
$c=5/3$ and $c=1.75$.

\subsubsection{Analysis of the competitive ratio}\label{sec:family21}

We start by showing two lemmas that will allow us to bound load on machines in $B$. 
Again, let time $n+1$ be the time when the entire job sequence  $\sigma=J_1, \ldots, J_n$ has
been scheduled and the migration phase starts. A job $J_i$, $1\leq i\leq n$, is 
{\em small at time $n+1$} if 
$p_i\leq (2c-3)L= (2c-3)L_n$; otherwise the job is {\em large at time $n+1$}. For any $M_j$, $1\leq j\leq m$,
let $\ell(j,n+1)$ be its load at time $n+1$ and let $\ell_s(j,n+1)$ be the load consisting of the jobs
that are small at time $n+1$. Let $L_{n+1}:=L$. 
\begin{lemma}\label{lem:2r1}
For any time $t$, $1\leq t\leq n+1$, and any $M_j\in B$, there holds $\ell(j,t)-p_l\leq (3-c)L_{t-1}$, 
where $J_l$ with $l<t$ is the last job assigned to $M_j$. 
\end{lemma}
\begin{proof}
By the definition of {\em ALG$(c)$}, when $J_l$ is assigned to $M_j$, all machines of $A$ have a load greater 
than $(c-1)L_l$ and $M_j$ is a least loaded machine in $B$. Hence $M_j$'s load at time $l$ is at most $(3-c)L_l$ 
since otherwise the  total load on the $m$ machines would be greater than 
$\lfloor m/2\rfloor (c-1) L_l+\lceil m/2\rceil (3-c) L_l
\geq mL_l\geq \sum_{i=1}^l p_i$, which is a contradiction. Hence $\ell(j,t) = \ell(j,l)+p_l \leq (3-c)L_l+p_l
\leq (3-c)L_{t-1}+p_l$.
\end{proof}

\begin{lemma}\label{lem:2r2}
Suppose that there exists a machine $M_{j^*}\in A$ with $\ell_s(j^*,n+1)< (2-c)L$. Then, for any
$M_j \in B$, $\ell(j,n+1)-p_l \leq (c-1)L$, where $J_l$ is the last job assigned to $M_j$. 
\end{lemma}
\begin{proof}
Consider any $M_j\in B$ and let $J_l$ be the last job assigned to it. First assume that $J_l$ is large
at time $l$. By the definition of {\em ALG$(c)$}, at time $l$ all machines of $A$ have a load greater
than $(3-c) L_l$. Moreover, $M_j$ is a least loaded machine in $B$ at time $l$. We argue that a least loaded 
machine in $B$ has a load of at most $(c-1) L_l$. If this were not the case, then immediately after the
assignment of $J_l$ the total load on the $m$ machines would be greater than
$\lfloor m/2\rfloor (3-c) L_l+\lceil m/2\rceil (c-1) L_l +p_l \geq (m/2-1/2) (3-c)L_l + (m/2+1/2) (c-1) L_l
+(2c-3)L_l = mL_l + (3c-5)L_l$. The inequality holds because $3-c \geq c-1$. Since $c \geq 5/3$ it follows
$\lfloor m/2\rfloor (3-c) L_l+\lceil m/2\rceil (c-1) L_l +p_l \geq mL_l  \geq \sum_{i=1}^l p_i$, which is 
a contradiction. Hence $\ell(j,n+1) = \ell(j,l)+p_l \leq (c-1)L_l+p_l \leq (c-1)L+p_l$.

Next assume that $J_l$ is small at time $l$. This implies $\ell_s(j,l) > (c-1)L_l$, for all $M_j\in A$. In particular,
$\ell_s(j^*,l) > (c-1)L_l$. Since $\ell_s(j^*,l)\leq \ell_s(j^*,n+1)< (2-c) L$ it follows
$L_l < (2-c)/(c-1)\cdot L$. By Lemma~\ref{lem:2r1}, $\ell(j,l+1)\leq (3-c)L_l+p_l$ and we conclude 
$\ell(j,n+1)=  \ell(j,l+1)\leq (3-c)L_l+p_l\leq (3-c)(2-c)/(c-1)\cdot L+p_l \leq (c-1) L+p_l$. 
The last inequality holds because $(3-c)(2-c)/(c-1)\leq c-1$ holds since $c\geq 5/3$.
\end{proof}

We next analyze the job migration phase assuming that the job removal step has already
taken place, i.e.\ each machine of $A$ has a load of at most $(c-1)L$ and the largest job was removed
from each machine of $B$. We show that given such a machine configuration each job of 
$R$ can be assigned to a machine so that a load bound of $cL$ is preserved. 
For the analysis of the reassignment step we study two cases depending on whether or not at time
$n+1$ all machines $M_j\in A$ have a load $\ell_s(j,n+1)\geq (2-c) L$.

\begin{lemma}\label{lem:2rem1}
If $\ell_s(j,n+1)\geq (2-c) L$, for all $M_j\in A$, then in the reassignment step all jobs of $R$ are
scheduled so that the resulting load on any of the machines is at most $c L$. 
\end{lemma}
\begin{proof}
By assumption, at the end of the job arrival phase $\ell_s(j,n+1)\geq (2-c)L$, for all $M_j\in A$. 
We first show that this property is maintained throughout the job removal step. Suppose that a job
$J_i$ that is small at time $n+1$ is removed from a machine $M_j\in A$. Since {\em ALG$(c)$} always
removes the largest jobs from a machine, $M_j$ currently contains no jobs that are large at time
$n+1$. Hence $M_j$'s current load $\ell(j)$ is equal to its current load $\ell_s(j)$ consisting
of jobs that are small at time $n+1$. Since a job removal needs to be performed,  $\ell_s(j) = \ell(j)
> (c-1) L$. Since $p_i \leq (2c-3) L$, the removal of $J_i$ leads to a load consisting of small jobs
of at least $\ell_s(j) -p_l > (c-1) L - (2c-3) L = (2-c)L$.

After the job removal step each machine $M_j\in A$ has a load of at most $(c-1) L$. By~Lemma~\ref{lem:2r1}
each machine of $B$ has a load of at most $(3-c) L<cL$ after {\em ALG$(c)$} has removed the largest
job from any of these machines.  We show that each $J_k\in R$
can be scheduled on a machine such that the resulting load is at most $c L$. Consider any
$J_k\in R$. There holds $p_k\leq L$. Suppose that $J_k$ cannot be feasibly scheduled on any of the
machines. Let $\ell(j)$ denote $M_j$'s load immediately before the assignment of $J_k$, $1\leq j\leq m$.
If $J_k$ cannot be placed on a machine in $A$, then each machine $M_j\in A$ must have a load greater
than $(c-1) L$: If $\ell(j)\leq (c-1) L$, then $\ell(j) + p_k\leq c L$ and the assignment of $J_k$ to
$M_j$ would be feasible. Hence since the start of the reassignment step each machine $M_j\in A$ must
have received at least one job $J_{i_j}$ and its current load is $\ell(j) \geq (2-c) L + p_{i_j}$.
When $J_{i_j}$ was reassigned, it could not be scheduled on any machine in $B$ without exceeding
a load of $c L$. This implies, in particular, that $\ell(\lfloor m/2\rfloor+j) + p_{i_j} > c L$.
Recall that the machines of $A$ are numbered $1, \ldots, \lfloor m/2\rfloor$ and those of $B$ are
numbered $\lfloor m/2\rfloor+1, \ldots, m$. Finally, since $J_k$ cannot be placed on a machine in
$B$, we have $\ell(m) + p_k > c L$.

It follows that when $J_k$ has to be scheduled the total processing time of the jobs is at least 
$$\sum_{j=1}^m \ell(j) + p_k \geq \lfloor m/2\rfloor (2-c) L + \sum_{j=1}^{\lfloor m/2\rfloor} p_{i_j} + 
\sum_{j=\lfloor m/2\rfloor+1}^m \ell(j) + p_k.$$
If $m$ is even, then $\sum_{j=\lfloor m/2\rfloor+1}^m \ell(j) = 
\sum_{j=1}^{m/2} \ell(m/2 +j)$. In this case we have
$$\sum_{j=1}^m \ell(j) + p_k \geq m/2 \cdot (2-c)L + \sum_{j=1}^{m/2} (\ell(m/2 +j) +p_{i_j}) + p_k
> m/2 \cdot (2-c) L + m/2\cdot c L \ = \ mL.$$ 
If $m$ is odd, then $\sum_{j=\lfloor m/2\rfloor+1}^m \ell(j) = 
\sum_{j=1}^{\lfloor m/2\rfloor} \ell(\lfloor m/2\rfloor +j)+ \ell(m)$
and 
\begin{eqnarray*}
\sum_{j=1}^m \ell(j) + p_k &\geq& \lfloor m/2\rfloor  \cdot (2-c) L  + \sum_{j=1}^{\lfloor m/2\rfloor} (\ell(\lfloor m/2\rfloor +j)+p_{i_j})
+ \ell(m)+ p_k\\
&> & \lfloor m/2\rfloor  \cdot (2-c) L  + \lfloor m/2\rfloor  \cdot c L +  c L\\
&=& (m/2 - 1/2)2L + c L \  > \ mL.
\end{eqnarray*}
In both cases with obtain $\sum_{i=1}^n p_i \geq  \sum_{j=1}^m \ell(j) + p_k > mL$, which contradicts the
definition of $L$.
\end{proof}

\begin{lemma}\label{lem:2rem2}
If $\ell_s(j^*,n+1)< (2-c) L$, for some $M_{j^*}\in A$, then in the reassignment step all jobs of $R$ are
scheduled so that the resulting load on any of the machines is at most $c L$. 
\end{lemma}

\begin{proof}
In the removal step {\em ALG$(c)$} removes the largest job from each machine $M_j\in B$. Hence, if
$\ell_s(j^*,n+1)< (2-c) L$ for some $M_j\in A$, then by Lemma~\ref{lem:2r2} each machine of $B$ has a load
of at most $(c-1)L$ after the removal step. Moreover, each machine of $A$ has a load of at most $(c-1)L$ after
the job removal.

Hence when the reassignment step starts, all machines have a load of at most $(c-1) L$. By the definition of $L$
each job has a processing time of at most $L$. Hence in the reassignment step the first $m$ jobs can be
scheduled without exceeding a load of $c L$ on any of the machines. {\em ALG$(c)$} sorts the jobs
of $R$ in order of non-increasing processing times. Thus when $m$ jobs of $R$ have been scheduled, each of
the remaining jobs has a processing time of at most $1/2 L$. This holds true because by the definition of
$L$ there cannot exist $m+1$ jobs of processing time greater than $1/2 L$. Each job of processing time
at most $1/2 L$ can be scheduled on a least loaded machine without exceeding a load of $c L$ since
$L+ 1/2 L < cL$. Hence every remaining job can be scheduled on a machine of $B$ and $A$. 
\end{proof}
Lemmas~\ref{lem:2rem1} and \ref{lem:2rem2} imply Theorem~\ref{th:family}.

\subsubsection{Analysis of the job migrations}\label{sec:family22}

It remains to evaluate the number of job removals in the job migration phase. We first consider
{\em ALG$(5/3)$}.

\begin{lemma}\label{lem:2r3}
In the removal step ALG$(5/3)$ removes at most seven jobs from each machine $M_j\in A$.
\end{lemma}
\begin{proof}
We show that, for any $M_j\in A$, it suffices to remove at most seven jobs from $M_j$ such that the
resulting load is upper bounded by $2/3L$. The lemma then follows because in each removal
operation {\em ALG$(5/3)$} removes the largest job.

First assume that $\ell_s(j,n+1)\leq 2/3L$. In this case it suffices to remove all jobs that are large 
at time $n+1$. Each such job has a processing time greater than $1/3L$ and was large at the time it was
assigned to $M_j$. Consider the last time when such a job was assigned to $M_j$. At that time $M_j$ had
a load of at most $4/3L$ and hence could contain no more than three jobs of processing time greater than
$1/3L$. Thus at time $n+1$ machine $M_j$ contains at most four of these large jobs.

Next assume $\ell_s(j,n+1)> 2/3L$. If $\ell_s(j,n)\leq  2/3L_n$, then $J_n$ is assigned to $M_j$ because 
$L=L_n$. Hence it suffices to remove $J_n$ and, as shown in the last paragraph, four additional jobs of
processing time greater than $1/3 L_n = 1/3L$.

In the following we concentrate on the case that $\ell_s(j,n+1)> 2/3L$ and $\ell_s(j,n)> 2/3L_n$.
Let $t^*$ be the earliest time such that  $\ell_s(j,t)> 2/3L_t$ holds for all times $t\geq t^*$. We
have $t^*>1$ because $\ell_s(j,0)=0$. We partition the jobs that reside on $M_j$ at time $n+1$ into
three sets. Set $T_1$ (set $T_2$) contains those jobs that were assigned to $M_j$ at or before time $t^*-1$
are small (large) at time $t^*-1$. Set $T_3$ contains the remaining jobs, which have arrived at or
after time $t^*$.

\begin{enumerate}[\bf { Claim~\ref{lem:2r3}.}1.]
\item Each job of $T_2\cup T_3$ is large at the time it is assigned to $M_j$.\label{c:2r31}
\item There holds $\sum_{J_i\in T_1\setminus\{J_l\}} p_i \leq 2/3 L_{t^*-1}$, where $J_l$ is the 
 job of $T_1$ that was assigned last to $M_j$.\label{c:2r32}
\item There holds $|T_2| \leq 4$.\label{c:2r33}
\item For any $J_l\in T_3$, $M_j$'s load immediately before the assignment of $J_l$ is at most $4/3 L_l$.\label{c:2r34}
\item Let $J_l\in T_3$ be the last job assigned to $M_j$. If $M_j$ contains at least four jobs, 
 different from $J_l$, each having a processing time of at least $1/6 L$, then it suffices
 to remove these four jobs and $J_l$ such that $M_j$'s resulting load is upper bounded by
 $2/3 L$. \label{c:2r35}
\item If there exists a $J_l\in T_3$ with $p_l<1/6L$, then $M_j$'s load immediately before
the assignment of $J_l$ is at most $2/3L$.\label{c:2r36}  
\item If there exists a $J_k\in T_2$ with $p_k<1/6L$, then $\sum_{J_i\in T_1}p_i + p_k\leq
2/3L$.\label{c:2r37}  
\end{enumerate}

{\em Proof of Claim~\ref{lem:2r3}.\ref{c:2r31}.} The jobs of $T_2$ are large at time $t^*-1$ and hence
at the time they were assigned to $M_j$. By the definition of $t^*$, $\ell_s(j,t)> 2/3 L_t$,
for any $t^*\leq t \leq n$, and hence {\em ALG($5/3$)} does not assign small jobs to $M_j$.

{\em Proof of Claim~\ref{lem:2r3}.\ref{c:2r32}.} By the choice of $t^*$, all jobs of $T_1\setminus\{J_l\}$ 
are small at time $t^*-1$ and their total processing time is at most $\ell_s(j,t^*-1)\leq 2/3 L_{t^*-1}$.

{\em Proof of Claim~\ref{lem:2r3}.\ref{c:2r33}.} Each job of $T_2$ has a processing time greater
than $1/3 L_{t^*-1}$. Consider the last time $l$ when a job $J_l\in T_2$ was assigned to $M_j$. 
Immediately before the assignment, $M_j$ had a load of at most $4/3 L_{t^*-1}$ and hence could contain 
not more than three jobs of processing time greater than $1/3 L_{t^*-1}$.

{\em Proof of Claim~\ref{lem:2r3}.\ref{c:2r34}.} Consider any $J_l\in T_3$. By 
Claim~\ref{lem:2r3}.\ref{c:2r31} $J_l$ is large at time $l$ and hence $M_j$'s load prior 
to the assignment of $J_l$ is at most $4/3 L_l$. 

{\em Proof of Claim~\ref{lem:2r3}.\ref{c:2r35}.} By Claim~\ref{lem:2r3}.\ref{c:2r34} $M_j$'s
load immediately before the assignment of $J_l$ is at most $4/3L_l$. Removing four jobs of
processing time at least $1/6L$ each as well as $J_l$ reduces $M_j$'s load to a value of at most 
$2/3 L$.

{\em Proof of Claim~\ref{lem:2r3}.\ref{c:2r36}.} By Claim~\ref{lem:2r3}.\ref{c:2r31} $J_l$ is large 
at time $l$ and hence $p_l> 1/3 L_l$. Since $p_l<1/6L$, we have $L_l<1/2L$.
By Claim~\ref{lem:2r3}.\ref{c:2r34}, $M_j$'s load immediately before the assignment of $J_l$ is at most 
$4/3 L_l$ and hence at most $2/3L$.

{\em Proof of Claim~\ref{lem:2r3}.\ref{c:2r37}.} Job $J_k$ is large at time $t^*-1$ and hence 
$p_k > 1/3 L_{t^*-1}$. Since $p_k < 1/6 L$ it follows $L_{t^*-1} < 1/2 L$. By Claim~\ref{lem:2r3}.\ref{c:2r32},
we have $\sum_{J_i\in T_1} p_i \leq 2/3 L_{t^*-1} + p_l$, where $J_l$ is the last job of $T_1$
assigned to $M_j$. Since $p_l$ is small at time $t^*-1$ we have $p_l \leq 1/3 L_{t^*-1} < 1/6 L$. 
In summary $\sum_{J_i\in T_1} p_i + p_k \leq 1/3L + 1/6 L + 1/6 L = 2/3 L$.

We proceed with the actual proof and distinguish two cases.

{\bf Case 1:} If $|T_2\cup T_3|\leq 4$, then by Claim~\ref{lem:2r3}.\ref{c:2r32} it suffices to remove the
jobs of $T_2\cup T_3$ and the last job of $T_1$ assigned to $M_j$. 

{\bf Case  2:} Assume $|T_2\cup T_3|\geq 5$. Then by Claim~\ref{lem:2r3}.\ref{c:2r33} there holds $|T_2| \leq 4$
and thus $T_3\neq \emptyset$. Let $J_l$ be the last job of $T_3$ assigned to $M_j$. If 
$T_2\cup T_3\setminus \{J_l\}$ contains at least four jobs of processing time at least $1/6 L$, then by 
Claim~\ref{lem:2r3}.\ref{c:2r35} it suffices to remove these four jobs and $J_l$. So suppose that this is
not the case. Then $T_2\cup T_3\setminus \{J_l\}$ must contain a job of processing time smaller than $1/6 L$.

Assume there exists a job in $T_3\setminus \{J_l\}$ with this property. Then let $J_{l'}$ be the last job 
assigned to $M_j$ having a processing time smaller than $1/6 L$. By Claim~\ref{lem:2r3}.\ref{c:2r36}, 
immediately before the assignment of $J_{l'}$ machine $M_j$ has a load of at most $2/3 L$. Therefore it
suffices to remove $J_{l'}$ and the jobs of $T_3$ subsequently scheduled on $M_j$. In addition to $J_l$,
this sequence consists of at most three jobs $J_k\neq J_l$, because $T_3\setminus \{J_l\}$ contains less 
than four jobs of processing time at least $1/6 L$. 

Finally consider the case that all jobs of $T_3\setminus \{J_l\}$ have a processing time of at least
$1/6 L$ and there is a job $J_{l'}\in T_2$ having a processing time smaller than $1/6 L$. By 
Claim~\ref{lem:2r3}.\ref{c:2r37} it suffices to remove  $T_2\setminus \{J_{l'}\}\cup T_3$. By 
Claim~\ref{lem:2r3}.\ref{c:2r33} we have $|T_2\setminus \{J_{l'}\}| \leq 3$. Since $T_3\setminus \{J_l\}$
contains less than four jobs, each having a processing time of at least $1/6 L$, we have $|T_3|\leq 4$.
We conclude that at most seven jobs have to be removed.
\end{proof}

Lemma~\ref{lem:2r2} ensures that in the job removal step {\em ALG$(5/3)$}
removes at most $7$ jobs from any machine in $A$. For any machine in $B$, one job is removed.
Hence the total number of migrations is at most $7\lfloor m/2 \rfloor + \lceil m/2 \rceil \leq 4m$.
This concludes the proof of Theorem~\ref{th:3}.
We next turn to the algorithm {\em ALG$(1.75)$}.
\begin{lemma}\label{lem:3r3}
In the job removal step  ALG$(1.75)$ removes at most four jobs from each machine $M_j\in A$. 
\end{lemma}
\begin{proof}
We show that, for any $M_j\in A$, it suffices to remove at most four jobs from $M_j$ such that the
resulting load is upper bounded by $0.75L$. 

First assume that $\ell_s(j,n+1)\leq 0.75L$. Then it suffices to remove all jobs that are large 
at time $n+1$. Each such job has a processing time greater than $0.5L$ and was large at the time it was
assigned to $M_j$. Consider the last time when such a job was assigned to $M_j$. At that time $M_j$ had
a load of at most $1.25L$ and hence could contain no more than two jobs of processing time greater than
$0.5L$. Thus at time $n+1$ machine $M_j$ contains at most three of these large jobs.

Next assume $\ell_s(j,n+1)> 0.75L$. If $\ell_s(j,n)\leq  0.75 L_n$, then $J_n$ is assigned to $M_j$ because 
$L=L_n$. Hence it suffices to remove $J_n$ and, as shown in the last paragraph, three additional jobs of
processing time greater than $0.5 L_n = 0.5L$.

We concentrate on the case that $\ell_s(j,n+1)> 0.75L$ and $\ell_s(j,n)> 0.75L_n$.
Let $t^*$ be the earliest time such that  $\ell_s(j,t)> 0.75L_t$ holds for all times $t\geq t^*$. 
We partition the jobs that reside on $M_j$ at time $n+1$ into
three sets. Set $T_1$ (set $T_2$) contains those jobs that were assigned to $M_j$ at or before time $t^*-1$
are small (large) at time $t^*-1$. Set $T_3$ contains the remaining jobs, which have arrived at or
after time $t^*$.
\begin{enumerate}[\bf { Claim~\ref{lem:3r3}.}1.]
\item Each job of $T_2\cup T_3$ is large at the time it is assigned to $M_j$.\label{c:3r31}
\item There holds $\sum_{J_i\in T_1\setminus\{J_l\}} p_i \leq 0.75 L_{t^*-1}$, where $J_l$ is the 
 job of $T_1$ that was assigned last to $M_j$.\label{c:3r32}
\item There holds $|T_2| \leq 3$.\label{c:3r33}
\item For any $J_l\in T_3$, $M_j$'s load immediately before the assignment of $J_l$ is at most $1.25 L_l$.\label{c:3r34}
\item Let $J_l\in T_3$ be the last job assigned to $M_j$. If $M_j$ contains at least three jobs, 
 different from $J_l$, each having a processing time of at least $1/6 L$, then it suffices
 to remove these three jobs and $J_l$ such that $M_j$'s resulting load is upper bounded by
 $0.75 L$. \label{c:3r35}
\item If there exists a $J_l\in T_3$ with $p_l<1/6L$, then $M_j$'s load immediately after
the assignment of $J_l$ is at most $0.75L$.\label{c:3r36}  
\item If $T'_2\subseteq T_2$ is a subset with $1\leq |T_2'|\leq 2$ and $p_i\leq 1/6L$, for all  $J_i\in T_2$, 
then $\sum_{J_i\in T_1}p_i + \sum_{J_i\in T'_2}p_i\leq 0.75L$.\label{c:3r37}  
\end{enumerate}

{\em Proof of Claim~\ref{lem:3r3}.\ref{c:3r31}.} The jobs of $T_2$ are large at time $t^*-1$ and hence
at the time they were assigned to $M_j$. By the definition of $t^*$, $\ell_s(j,t)> 0.75 L_t$,
for any $t^*\leq t \leq n$, and hence {\em ALG($1.75$)} does not assign small jobs to $M_j$
at times $t\geq t^*$.

{\em Proof of Claim~\ref{lem:3r3}.\ref{c:3r32}.} All jobs of $T_1\setminus\{J_l\}$ 
are small at time $t^*-1$ and their total processing time is at most $\ell_s(j,t^*-1)\leq 0.75 L_{t^*-1}$,
by the choice of $t^*$. 

{\em Proof of Claim~\ref{lem:3r3}.\ref{c:3r33}.} Each job of $T_2$ has a processing time greater
than $0.5 L_{t^*-1}$. Consider the last time $l$ when a job $J_l\in T_2$ was assigned to $M_j$. 
Immediately before the assignment, $M_j$ had a load of at most $1.25 L_{t^*-1}$ and hence could contain 
not more than two jobs of processing time greater than $0.5 L_{t^*-1}$.

{\em Proof of Claim~\ref{lem:3r3}.\ref{c:3r34}.} Consider any $J_l\in T_3$. By 
Claim~\ref{lem:3r3}.\ref{c:2r31} $J_l$ is large at time $l$ and hence $M_j$'s load prior 
to the assignment of $J_l$ is at most $1.25 L_l$. 

{\em Proof of Claim~\ref{lem:3r3}.\ref{c:3r35}.} By Claim~\ref{lem:3r3}.\ref{c:3r34} $M_j$'s
load immediately before the assignment of $J_l$ is at most $1.25 L_l$. Removing three jobs of
processing time at least $1/6L$ each as well as $J_l$ reduces $M_j$'s load to a value of at most 
$0.75 L$.

{\em Proof of Claim~\ref{lem:3r3}.\ref{c:3r36}.} By Claim~\ref{lem:3r3}.\ref{c:3r31} $J_l$ is large 
at time $l$ and hence $p_l> 0.5 L_l$. Since $p_l<1/6L$, we have $L_l<1/3L$.
Using Claim~\ref{lem:3r3}.\ref{c:3r34} we obtain that $M_j$'s load immediately after the assignment 
of $J_l$ is at most $1.25 L_l + p_l \leq 5/12 L + 1/6L < 0.75L$.

{\em Proof of Claim~\ref{lem:3r3}.\ref{c:3r37}.} Any job $J_i\in T'_2$ is large at time $t^*-1$ and hence 
$p_i > 0.5 L_{t^*-1}$. Since $p_i < 1/6 L$ it follows $L_{t^*-1} < 1/3 L$. By Claim~\ref{lem:3r3}.\ref{c:3r32},
we have $\sum_{J_i\in T_1} p_i \leq 0.75 L_{t^*-1} + p_l\leq 0.25L + 1/6 L$, where $J_l$ is the last job of $T_1$
assigned to $M_j$. Thus $\sum_{J_i\in T_1}p_i + \sum_{J_i\in T'_2}p_i\leq 0.25L + 3\cdot 1/6L\leq 0.75L$.

We finish the proof of the lemma using a case distinction on the size of $T_3$.

\begin{itemize}
\item $|T_3|=0$: Then by Claim~\ref{lem:3r3}.\ref{c:3r32} it suffices to remove $T_2$ and the last job
of $T_1$ assigned to $M_j$. By Claim~\ref{lem:3r3}.\ref{c:3r33}, $T_2$ contains no more than three jobs.

\item $|T_3|= 1$: We may assume that the only job $J_l\in T_3$ has a processing time of at least
$1/6L$ since otherwise by Claim~\ref{lem:3r3}.\ref{c:3r36} no job has to be removed. Moreover, we may
assume that $|T_2|=3$ since otherwise, by Claim~\ref{lem:3r3}.\ref{c:3r32} it suffices to remove
$T_2\cup T_3$ and the last job of $T_1$ assigned to $M_j$. If all the jobs of $T_2$ have a processing
time of at least $1/6L$, then Claim~\ref{lem:3r3}.\ref{c:3r35} ensures that it suffices to remove
$T_2\cup T_3$. If one job in $T_2$ has a processing time of at most $1/6L$, then
Claim~\ref{lem:3r3}.\ref{c:3r37} ensures that it suffices to remove the other two jobs of $T_2$
and $T_3$.

\item $|T_3|= 2$: We assume that both jobs in $T_3$ have a processing time of at least $1/6 L$ 
since otherwise, by Claim~\ref{lem:3r3}.\ref{c:3r36}, we can just remove one job of $T_3$ and
$T_2$. If $|T_2|=1$, then by Claim~\ref{lem:3r3}.\ref{c:3r32} it suffices to remove $T_2\cup T_3$
and the last job of $T_1$ assigned to $M_j$. It remains to consider the case $|T_2|\geq 2$. If
none of the jobs in $T_2$ has a processing time smaller than $1/6L$, then Claim~\ref{lem:3r3}.\ref{c:3r35}
applies. If one of the jobs has a processing time smaller than $1/6L$, then Claim~\ref{lem:3r3}.\ref{c:3r37}
applies and it suffices to remove the at most two other jobs of $T_2$ and the jobs of $T_3$.

\item $|T_3|= 3$: Again we assume that all jobs in $T_3$ have a processing time of at least $1/6 L$ 
since otherwise the desired statement follows from Claim~\ref{lem:3r3}.\ref{c:3r36}, Moreover,
we assume $|T_2|>0$; otherwise we can apply again Claim~\ref{lem:3r3}.\ref{c:3r32}. If there
is one job in $T_2$ having a processing time of at least $1/6L$, the desired number of job removals
follows from Claim~\ref{lem:3r3}.\ref{c:3r35}. If this is not the case, then Claim~\ref{lem:3r3}.\ref{c:3r37}
ensures that it suffices to remove the last job of $T_2$ assigned to $M_j$ as well as $T_3$.

\item $|T_3|\geq 4$: If four jobs in $T_3$ have a processing time of at least $1/6 L$, then by
Claim~\ref{lem:3r3}.\ref{c:3r35} it is sufficient to remove three out of these in addition to the last job assigned to $M_j$. If at most three jobs have 
a processing time of at least $1/6 L$, then let $J_l\in T_3$ be last jobs assigned to $M_j$ having
a processing time smaller than  $1/6 L$. By Claim~\ref{lem:3r3}.\ref{c:3r36} it suffices to remove
the jobs of $T_3$ subsequently assigned to $M_j$, and there exist at most three of these.
\end{itemize}
This concludes the proof.
\end{proof}

Recall that {\em ALG$(1.75)$} migrates $\lceil m/2\rceil$ jobs from machines in $B$. Hence, using the 
above Lemma~\ref{lem:3r3}, we obtain that the total number of migrations is at most 
$4\lfloor m/2\rfloor + \lceil m/2\rceil\leq 2.5m$. 
This finishes the proof of Theorem~\ref{th:4}.

\section*{Appendix}

\begin{proof}[Proof of Lemma~\ref{lem:l1}]
Fix $m\geq 2$. We first evaluate $f_m(2)$ and $f_m(1+1/(3m))$. For $\alpha=2$, we have 
$\lceil(1-1/\alpha)m\rceil \geq m/2$. Hence $\lceil(1-1/\alpha)m\rceil \alpha/m \geq 1$ and $f_m(2) \geq 1$.
For $\alpha=1+ 1/(3m)$, there holds $\lceil(1-1/\alpha)m\rceil =1$. Thus $f_m(1+1/(3m)) = 1/(3m) H_{m-1}
+1/m + 1/(3m^2)< 1/3 + 1/2 + 1/12 < 1$. It remains to show that $f_m(\alpha)$ is continuous and strictly
increasing. To this end we show that, for any $\alpha>1$ and small $\epsilon>0$, 
$f_m(\alpha+\epsilon) - f_m(\alpha)$ is strictly positive and converges to~0 as $\epsilon\rightarrow 0$.

First consider an $\alpha >1$ such that $(1-1/\alpha)m\notin \mathbb{N}$. In this case we choose 
$\epsilon > 0$ such that $\lceil(1-1/(\alpha+\epsilon))m\rceil= \lceil(1-1/\alpha)m\rceil$. We have
\begin{eqnarray*}
f_m(\alpha) &=& (\alpha-1)(H_{m-1}-H_{\lceil(1-1/\alpha)m\rceil-1}) + \lceil(1-1/\alpha)m\rceil \alpha/m\\
f_m(\alpha+\epsilon) &=& (\alpha+\epsilon-1)(H_{m-1}-H_{\lceil(1-1/\alpha)m\rceil-1}) + \lceil(1-1/\alpha)m\rceil (\alpha+\epsilon)/m.
\end{eqnarray*}
Thus $f_m(\alpha+\epsilon) - f_m(\alpha) = \epsilon (H_{m-1}-H_{\lceil(1-1/\alpha)m\rceil-1}) + 
\lceil(1-1/\alpha)m\rceil \epsilon/m$. Since $\alpha >1$ there holds 
$\lceil(1-1/\alpha)m\rceil \geq 1$ and thus $f_m(\alpha+\epsilon) - f_m(\alpha) >0$.
Moreover, $f_m(\alpha+\epsilon) - f_m(\alpha)$ tends to~0 as $\epsilon\rightarrow 0$.

Next let $\alpha >1$ such that $(1-1/\alpha)m\in \mathbb{N}$. In this case we choose 
$\epsilon > 0$ such that $\lceil(1-1/(\alpha+\epsilon))m\rceil= \lceil(1-1/\alpha)m\rceil +1$.
There holds
\begin{eqnarray*}
f_m(\alpha) &=& (\alpha-1)(H_{m-1}-H_{\lceil(1-1/\alpha)m\rceil-1}) + \lceil(1-1/\alpha)m\rceil \alpha/m\\
f_m(\alpha+\epsilon) &=& (\alpha+\epsilon-1)(H_{m-1}-H_{\lceil(1-1/\alpha)m\rceil}) + (\lceil(1-1/\alpha)m\rceil+1) (\alpha+\epsilon)/m.
\end{eqnarray*}
Taking into account that $(1-1/\alpha)m\in \mathbb{N}$ we obtain
\begin{eqnarray*}
f_m(\alpha+\epsilon) - f_m(\alpha) &=& -(\alpha-1)\cdot 1/((1-1/\alpha)m) + 
\epsilon (H_{m-1}-H_{\lceil(1-1/\alpha)m\rceil})\\
& &  + (\lceil(1-1/\alpha)m\rceil+1) \epsilon/m + \alpha/m\\
&= &\epsilon (H_{m-1}-H_{\lceil(1-1/\alpha)m\rceil}) + (\lceil(1-1/\alpha)m\rceil+1) \epsilon/m.
\end{eqnarray*}
Again, $f_m(\alpha+\epsilon) - f_m(\alpha)$ is strictly positive and tends to~0 as $\epsilon\rightarrow 0$. 
\end{proof}

\begin{proof}[Proof of Lemma~\ref{lem:l2}]
We first prove that $(\alpha_m)_{m\geq 2}$ is non-decreasing. A first observation is that
$\alpha_m \leq m$ because $f_m(m)\geq 1$. We will show that, for any $m\geq 3$ and
$1<\alpha\leq m$, there holds $f_{m-1}(\alpha) \geq f_m(\alpha)$. This implies 
$1= f_{m-1}(\alpha_{m-1}) \geq f_m(\alpha_{m-1})$. By Lemma~\ref{lem:l1}, $f_m$ is strictly increasing
and thus $\alpha_m\geq \alpha_{m-1}$. Consider a fixed $\alpha$ with $1<\alpha\leq m$. We study two cases 
depending on whether or not $\lceil(1-1/\alpha)(m-1)\rceil = \lceil(1-1/\alpha)m\rceil$.

If $\lceil(1-1/\alpha)(m-1)\rceil = \lceil(1-1/\alpha)m\rceil$, then 
\begin{eqnarray*}
f_m(\alpha) &=& (\alpha-1)(H_{m-1}-H_{\lceil(1-1/\alpha)m\rceil-1}) + \lceil(1-1/\alpha)m\rceil \alpha/m\\
f_{m-1}(\alpha) &=& (\alpha-1)(H_{m-2}-H_{\lceil(1-1/\alpha)m\rceil-1}) + \lceil(1-1/\alpha)m\rceil \alpha/(m-1).
\end{eqnarray*}
We obtain $f_{m-1}(\alpha) - f_m(\alpha) = -(\alpha-1)/(m-1) +  \lceil(1-1/\alpha)m\rceil \alpha/(m(m-1))
\geq -(\alpha-1)/(m-1) + (\alpha-1)/(m-1) =0$ and thus $f_{m-1}(\alpha) \geq f_m(\alpha)$.

If $\lceil(1-1/\alpha)(m-1)\rceil < \lceil(1-1/\alpha)m\rceil$, then 
$\lceil(1-1/\alpha)(m-1)\rceil = \lceil(1-1/\alpha)m\rceil-1$ and
\begin{eqnarray*}
f_m(\alpha) &=& (\alpha-1)(H_{m-1}-H_{\lceil(1-1/\alpha)m\rceil-1}) + \lceil(1-1/\alpha)m\rceil \alpha/m\\
f_{m-1}(\alpha) &=& (\alpha-1)(H_{m-2}-H_{\lceil(1-1/\alpha)m\rceil-2}) + (\lceil(1-1/\alpha)m\rceil-1) \alpha/(m-1).
\end{eqnarray*}
Since $\alpha >1$ there holds $\lceil(1-1/\alpha)(m-1)\rceil \geq 1$. Hence in our case 
$\lceil(1-1/\alpha)m\rceil \geq 2$ and $\lceil(1-1/\alpha)m\rceil -1 >0$. We obtain
$$\textstyle{f_{m-1}(\alpha) - f_m(\alpha) = -{\alpha-1\over m-1} + {\alpha-1\over \lceil(1-1/\alpha)m\rceil -1}
+ \lceil(1-1/\alpha)m\rceil {\alpha\over m(m-1)} - {\alpha\over m-1}.}$$
Choose $x$, with $0\leq x <1$, such that  $\lceil(1-1/\alpha)m\rceil = (1-1/\alpha)m +x$. Then 
\begin{eqnarray*}
\textstyle{f_{m-1}(\alpha) - f_m(\alpha)} &=& \textstyle{ -{\alpha-1\over m-1} + {\alpha-1\over (1-1/\alpha)m +x-1}
+ (1-1/\alpha)m {\alpha\over m(m-1)} + {\alpha x\over m(m-1)}- {\alpha\over m-1}}\\
&=& \textstyle{{\alpha-1\over (1-1/\alpha)m +x-1} + {\alpha x\over m(m-1)}- {\alpha\over m-1}}
\end{eqnarray*}
In order to establish $f_{m-1}(\alpha) - f_m(\alpha)\geq 0$ is suffices to show 
$$\textstyle{{\alpha-1\over (1-1/\alpha)m +x-1} \geq {\alpha(m-x)\over m(m-1)}.}$$ 
This is equivalent to $(\alpha-1)m(m-1) \geq (m-x)((\alpha-1)m +\alpha x-\alpha)$. Standard algebraic
manipulation yield that this is equivalent to $m \geq mx - \alpha x^2+\alpha x$. Let 
$g(x) = mx - \alpha x^2+\alpha x$, for any real number $x$. This function is increasing for any
$x < (m+\alpha)/(2\alpha)$. Since $\alpha \leq m$, the function is increasing for any $x <1$. As
$g(0) = 0$ and $g(1) = m$, it follows that $m \geq mx - \alpha x^2+\alpha x$ holds for all
$0\leq x <1$. We conclude  $f_{m-1}(\alpha) - f_m(\alpha)\geq 0$.

It is easy to verify that $f_2(4/3)=1$. We show that $\lim_{m\rightarrow \infty} \alpha_m$ is upper 
bounded by $W_{-1}(-1/e^2)/(1+ W_{-1}(-1/e^2))$. Ces\'aro~\cite{C} proved
\begin{equation}
0 < H_m - \frac{1}{2} \ln \left (m(m+1) \right) -\gamma < \frac{1}{6m(m+1)}, \label{lbCesaro}
\end{equation}
where $\gamma \approx 0.577$ is the Euler-Mascheroni constant.
Using this inequality we find, for any $c$ with $0< c\leq 1$ and $\lceil cm\rceil -2 >0$,
\begin{eqnarray*}
H_{m-1} - H_{\lceil cm \rceil -2 } & >&  \frac{1}{2} \ln ((m-1)m) + \gamma -  \frac{1}{2} \ln ((\lceil cm \rceil-2)(\lceil cm\rceil-1)) -\gamma - \frac{1}{6(\lceil cm \rceil -2)(\lceil cm \rceil-1)}\\
& \geq &\frac{1}{2} \left ( \ln (m-1)+ \ln m  -   \ln (cm -1) - \ln( cm)      \right)   - \frac{1}{2(\lceil cm \rceil-1)} \\
& = &\frac{1}{2} \left ( \ln (m-1)+ \ln m -   \ln (c(m-1/c)) - \ln(cm)      \right)  - \frac{1}{2(\lceil cm \rceil-1)} \\
& = &\frac{1}{2} \left ( \ln (m-1) - \ln(m-1/c) -2 \ln (c)      \right) - \frac{1}{2(\lceil cm \rceil-1)} \\
& \geq & \frac{1}{2}  \left(   2 \ln (1/c) \right)   - \frac{1}{2(\lceil cm \rceil-1)} \\
&  \geq & \ln (1/c)    - \frac{1}{2(cm-1)}, 
\end{eqnarray*}
where the second to last inequality holds since $\ln (m-1/c) \leq \ln(m-1)$. for $0<c\leq1$ and sufficiently
large $m$.
We obtain
\begin{eqnarray*}
f_m(\alpha)&=&(\alpha-1)(H_{m-1}-H_{\lceil (1-1/\alpha)m \rceil-1}) + \left( \lceil (1-1/\alpha)m \rceil  \right) \frac{\alpha}{m} \\ 
& > &(\alpha-1)\left( \ln  (\frac{\alpha}{\alpha-1})    - \frac{1}{2((1-1/\alpha)m -1)} - \frac{1}{\lceil (1-1/\alpha)m \rceil-1} \right) + \left( \lceil (1-1/\alpha)m \rceil  \right) \frac{\alpha}{m} \\
& \geq & (\alpha-1)\left( \ln  (\frac{\alpha}{\alpha-1})     - \frac{1}{(1-1/\alpha)m -1} \right ) + \alpha-1 =: F(m).
\end{eqnarray*}
Obviously, $\lim_{m\to \infty} F(m) = (\alpha-1) \ln  (\frac{\alpha}{\alpha-1})  +\alpha-1$. We show that 
$(\alpha-1) \ln  (\frac{\alpha}{\alpha-1})  +\alpha-1 = 1$, for $\alpha=\frac{1}{1-\delta}$, where $\delta = -1/W_{-1}(-1/e^2)$.

Equation $(\alpha-1) \ln  (\frac{\alpha}{\alpha-1})  + \alpha-1 = 1$ is equivalent to 
$\ln  (\frac{\alpha}{\alpha-1})+1   = \frac{1}{\alpha-1}$, which in turn is equivalent to 
$$\frac{\alpha}{\alpha-1} \cdot e = e^\frac{1}{\alpha-1}.$$
Substituting $x=1/(\alpha-1)$, which is equivalent to $\alpha=1/x+1$, we find that the above is equivalent to 
$xe+e = e^x$. Applying the Lambert $W$ function we find that  $x=-W_{-1}(-1/e^2)-1$ is a solution of the former equality. Substituting we conclude that in fact $\alpha= W_{-1}(-1/e^2)/(1+ W_{-1}(-1/e^2))$ satisfies the equality. 
Using the same techniques we can show that $\lim_{m\rightarrow \infty} \alpha_m$ is lower bounded 
by $W_{-1}(-1/e^2)/(1+ W_{-1}(-1/e^2))$. In the calculations, (\ref{lbCesaro}) yields that
$H_{m-1} - H_{\lceil cm \rceil} < \ln(1/c) + 1/(2m)$.
\end{proof}

\end{document}